\documentclass[leqno]{article}
\pdfoutput=1
\usepackage{amsthm}
\usepackage{amsfonts}
\usepackage{amssymb}
\usepackage{amsmath}
\usepackage{xcolor}
\usepackage{tikz-cd}
\usepackage{appendix}
\usepackage{hyperref}
\usepackage{comment}
\usepackage[shortlabels]{enumitem}
\usepackage[capitalise]{cleveref}

\newtheorem{theorem}{Theorem}[section]
\newtheorem{lemma}[theorem]{Lemma}
\newtheorem{proposition}[theorem]{Proposition}
\newtheorem{corollary}[theorem]{Corollary}

\theoremstyle{definition}
\newtheorem{definition}[theorem]{Definition}

\newtheorem{example}[theorem]{Example}

\theoremstyle{remark}
\newtheorem{remark}[theorem]{Remark}

\newcommand{\A}{\mathcal A}
\newcommand{\B}{\mathcal B}

\newcommand{\F}{\mathcal F}
\newcommand{\G}{\mathcal G}

\newcommand{\M}{\mathcal M}
\newcommand{\N}{\mathcal N}
\newcommand{\Q}{\mathcal Q}
\newcommand{\R}{\mathcal R}
\renewcommand{\S}{\mathcal S}
\newcommand{\T}{\mathcal T}

\newcommand{\V}{\mathcal V}
\newcommand{\W}{\mathcal W}
\newcommand{\X}{\mathcal X}
\newcommand{\Y}{\mathcal Y}

\newcommand{\CC}{\mathbb C}

\usepackage[T3,T1]{fontenc}
\DeclareSymbolFont{tipa}{T3}{cmr}{m}{n}
\DeclareMathAccent{\bend}{\mathalpha}{tipa}{16}

\newcommand{\EV}{\quad \Longleftrightarrow \quad}
\newcommand{\iso}{\cong}
\newcommand{\Inc}{\mathrm{Inc}}
\newcommand{\Eval}{\mathrm{Eval}}
\newcommand{\Id}{\mathrm{Id}}
\renewcommand{\:}{\colon}
\newcommand{\cat}{\mathsf}
\newcommand{\Cl}{\mathrm{Cl}}
\newcommand{\id}{\mathrm{id}}
\newcommand{\one}{\mathbf{1}}
\newcommand{\tr}{\mathrm{tr}}
\newcommand{\boxprod}{\mathbin{\square}}
\newcommand{\At}{\mathrm{At}}

\allowdisplaybreaks

\title{Quantum graphs of homomorphisms}
\author{Andre Kornell and Bert Lindenhovius}
\date{}

\newcommand{\Addresses}{{
  \bigskip
  \footnotesize

  \noindent \textsc{Department of Mathematical Sciences, New Mexico State University}
  \par\nopagebreak
  \noindent Las Cruces, New Mexico, United States of America
  \par\nopagebreak
  \noindent \textit{E-mail address}: \texttt{kornell@nmsu.edu}

  \medskip

  \noindent \textsc{Institute for Mathematical Methods in Medicine and Data Based Modeling, Johannes Kepler University Linz}
  \par\nopagebreak
  \noindent Linz, Austria
  \par\nopagebreak
  \noindent \textit{E-mail address}: \texttt{albertus.lindenhovius@jku.at}

}}

\begin{document}

\maketitle

\begin{abstract}
We introduce a category $\cat{qGph}$ of quantum graphs, whose definition is motivated entirely from noncommutative geometry. For all quantum graphs $G$ and $H$ in $\cat{qGph}$, we then construct a quantum graph $[G,H]$ of homomorphisms from $G$ to $H$, making $\cat{qGph}$ a closed symmetric monoidal category. We prove that for all finite graphs $G$ and $H$, the quantum graph $[G,H]$ is nonempty iff the $(G,H)$-homomorphism game has a winning quantum strategy, directly generalizing the classical case.

The finite quantum graphs in $\cat{qGph}$ are tracial, real, and self-adjoint, and the morphisms between them are CP morphisms that are adjoint to a unital $*$-homomorphism. We observe that Weaver's two notions of a CP morphism coincide in this context. We also include a short proof that every finite reflexive quantum graph is the confusability quantum graph of a quantum channel.
\end{abstract}

\vspace{.5cm}


\vspace{.5cm}

\noindent{MSC 2020:
46L89, % Other “noncommutative” mathematics based on C*-algebra theory
05C76, % Graph operations (line graphs, products, etc.)
81P40, % Quantum coherence, entanglement, quantum correlations
81P47, % Quantum channels, fidelity 
18D15. % Closed categories (closed monoidal and Cartesian closed categories, etc.)

\section{Introduction}\label{section 1}}

\subsection{The graph of homomorphisms $G \to H$}

In the meaning of the present article, a \emph{graph} $G$ is a pair $(V_G, e_G)$ such that $V_G$ is a set and $e_g$ is a symmetric relation on $V_G$. The set $V_G$ may be infinite or empty, and the relation $e_g$ may be neither reflexive nor irreflexive. In other words, our graphs may have any number of vertices, and loops are allowed, but multiple edges are forbidden. We write $g_1 \sim g_2$ if $(g_1, g_2) \in e_G$.

Graphs form a category $\cat{Gph}$ in the obvious way: a \emph{homomorphism} $G \to H$ is a function $\phi\:V_G \to V_H$ such that $g_1 \sim g_2$ implies $\phi(g_1) \sim \phi(g_2)$. Establishing the elementary properties of $\cat{Gph}$ is a routine exercise. The initial object is the empty graph, which we notate $K_0$, and the terminal object is a vertex with a loop, which we notate $\overline K_1$. The coproduct of two graphs $G$ and $H$ is their disjoint union $G + H$, and their product $G \times H$ is defined by $V_{G \times H} = V_G \times V_H$ with $(g_1, h_1) \sim (g_2, h_2)$ if $g_1 \sim g_2$ and $h_1 \sim h_2$.

The set $\cat{Gph}(G,H)$ of all homomorphisms $G \to H$ is itself a graph in a natural way: $\phi \sim \psi$ if $\phi(g) \sim \psi(g)$ for all $g \in V_G$. This graph structure is closely related to a different product of graphs. The \emph{box product} $G \boxprod H$ is defined by $V_{G \boxprod H} = V_G \times V_H$ with $(g_1, h_1) \sim (g_2, h_2)$ if both $g_1 = g_2$ and $h_1 \sim h_2$ or both $g_1 \sim g_2$ and $h_1 = h_2$. This relationship is described by a universal property:
\begin{equation}\label{equation 1}\tag{$\ast$}
\begin{tikzcd}
K \boxprod G
\arrow{rrrd}{\phi}
\arrow[dotted]{d}[swap]{! \times \id}
&
&
&
\\
{\cat{Gph}(G,H)} \boxprod G
\arrow{rrr}[swap]{\mathrm{eval}}
&
&
&
H
\end{tikzcd}
\end{equation}
For each graph $K$, each homomorphism $\phi\: K \boxprod G \to H$ factors uniquely through the evaluation homomorphism $\cat{Gph}(G,H) \boxprod G \to H$ as in eq.~\ref{equation 1} (cf.~\cite{zbMATH07883979}).

\subsection{The quantum graph of homomorphisms $G \to H$}

The immediate purpose of the present article is to establish a variant of this result for quantum graphs in the context of noncommutative geometry \cite{zbMATH01452515, zbMATH07856619}. A few relevant definitions of quantum graphs appear in the literature \cite{zbMATH06008057, zbMATH06936038, zbMATH07202497}. They are inequivalent but closely related, and they ultimately originate in the problem of error correction for quantum information channels \cite{zbMATH06727895}. Intuitively, quantum graph structure encodes the possibility that two states may be confused with each other, e.g., in the sense of having a nonzero transition probability after transmission through a quantum channel.

Our approach adheres to orthodox noncommutative geometry. We generalize sets to discrete quantum spaces \cite{zbMATH04152742}, i.e., quantum sets \cite{zbMATH07287276}, and generalize relations to quantum relations \cite{zbMATH07388954}, which correspond to projections in this setting \cite[Theorems~3.6~and~B.8]{zbMATH07287276}. Explicitly, in the meaning of the present article, a \emph{quantum graph} is essentially a von Neumann algebra $\M \subseteq \B(H)$ and an ultraweakly closed operator system $\R \subseteq \B(H)$ such that $\M' \cdot \R \cdot \M' \subseteq \R$ and $\R^\dagger = \R$ (cf.~\cite[Definition~2.6(d)]{zbMATH07388954}) and such that $\M \iso \bigoplus_{i \in I}^{\ell^\infty} M_{n_i}(\CC)$ for some indexed family of positive integers $\{n_i\}_{i \in I}$, which may be infinite.

We define a category $\cat{qGph}$ of quantum graphs, which contains $\cat{Gph}$ as a full subcategory. We use the generalization of the box product that is natural to this setting \cite[Definition~2.1]{arXiv:2408.11911}. For each pair of quantum graphs $G$ and $H$, we construct the quantum graph $[G, H]$ and prove \cref{D}, establishing an analogue of the universal property in eq.~\ref{equation 1}:

\begin{equation*}
\begin{tikzcd}
K \boxprod G
\arrow{rrrd}{\Phi}
\arrow[dotted]{d}[swap]{! \times \Id}
&
&
&
\\
{[G,H]} \boxprod G
\arrow{rrr}[swap]{\Eval}
&
&
&
H
\end{tikzcd}
\end{equation*}
Thus, $[G,H]$ acts as the quantum graph of all homomorphisms $G \to H$ from the perspective of noncommutative geometry.

To support this interpretation, we prove \cref{J}, providing a connection to graph homomorphism games \cite{zbMATH06555288}. Explicitly, we prove that the following conditions are equivalent for all finite simple graphs $G$ and $H$:
\begin{enumerate}
\item there is a winning quantum strategy for the $(G,H)$-homomorphism game,
\item the quantum graph $[G, H]$ is not empty.
\end{enumerate}
Because there exists a winning strategy for the $(G,H)$-homomorphism game iff there exists a homomorphism $G \to H$ \cite[section~2.1]{zbMATH06555288}, it is said that \emph{there is a quantum homomorphism} $G \to H$ if there exists a winning quantum strategy for this game \cite[section~1]{zbMATH06555288}. Thus, $[G,H]$ acts as the quantum graph of all homomorphisms $G \to H$ from the perspective of quantum information theory.

In short, the $(G,H)$-homomorphism game exhibits quantum nonlocality \cite{PhysicsPhysiqueFizika.1.195} iff the quantum graph $[G,H]$ is nonempty but the graph $\cat{Gph}(G,H)$ is empty. Man\v{c}inska and Roberson have shown that this occurs for graphs $G$ and $H$ on fourteen and four vertices, respectively \cite{oddities}. Harris has also recently shown that graph homomorphism games are universal for synchronous games \cite{zbMATH07924597}, which form a prominent subclass of nonlocal games.

\subsection{Comments}

Our construction of $[G,H]$ combines two threads in noncommutative geometry. The first of these threads is the construction of quantum function spaces that began with the spaces found by Woronowicz \cite{zbMATH03810280,zbMATH05539479} in the topological setting, i.e., in the setting of $C^*$-algebras. There, it is assumed that the domain is a finite quantum space, i.e., that the corresponding $C^*$-algebra is finite-dimensional, and that the codomain is a finite-dimensional compact quantum space, i.e., that the corresponding $C^*$-algebra is finite-dimensional and unital.

Quantum function spaces were then constructed by the first author in the setting of von Neumann algebras \cite{zbMATH06813783} and hereditarily atomic von Neumann algebras \cite{zbMATH07287276} without assumptions on the domain and the codomain. The latter construction produces a quantum set of functions $\Y^\X$ for all quantum sets $\X$ and $\Y$ because hereditarily atomic von Neumann algebras are a formalization of quantum sets (section~\ref{section 2}). This quantum set of functions was used to construct a quantum poset of monotone functions in \cite{zbMATH07605379}, and the same approach is used in the present article to construct the quantum graph $[G,H]$.

The second thread that leads to our construction of $[G,H]$ is the analysis of quantum strategies for nonlocal games. This analysis began with a criterion for the existence of winning quantum strategies \cite[Proposition~1]{zbMATH05540859}, which was then gradually understood to refer to generalized homomorphisms \cite{zbMATH06694253,zbMATH06552074,zbMATH07204369,zbMATH06936038}. The first author then observed that, from the perspective of noncommutative geometry, these quantum homomorphisms should be understood as \emph{quantum families} of homomorphisms \cite{zbMATH07287276}, anticipating \cref{J}.

Early in this development, Oritz and Paulsen characterized the existence of winning quantum strategies in terms of a universal construction \cite{zbMATH06552074}, introducing another aspect of \cref{J}. Specifically, they proved that there is a winning quantum strategy for the $(G,H)$-homomorphism game iff its $C^*$-algebra $\A(G,H)$ is nonzero \cite[Theorem~4.7(3)]{zbMATH06552074}. This algebraic approach has yielded significant progress in our understanding of nonlocal games, e.g., \cite{zbMATH07202497, zbMATH07924597}. It is natural to speculate that the atoms of the vertex quantum set $\V_{[G,H]}$ are in one-to-one correspondence with the irreducible finite-dimensional representations of $\A(G,H)$, but we do not pursue this correspondence here because it is peripheral to the major themes of this article.

The novelty of \cref{J} is that, in contrast to the bespoke $C^*$-algebra $\A(G,H)$, the quantum graph $[G,H]$ arises directly from the first principles of noncommutative geometry, quite apart from any considerations involving the graph homomorphism game. Indeed, it is uniquely determined by the natural one-to-one correspondence between homomorphisms $K \to [G,H]$ and homomorphisms $K \boxprod G \to H$ (\cref{{D}}), so in effect, it is uniquely determined by the well-established notion of a discrete quantum space \cite{zbMATH04152742}. We refer to these discrete quantum spaces as quantum sets (section~\ref{section 2}).

Furthermore, our construction of $[G,H]$ proceeds entirely in terms of these quantum sets as objects in their own right rather than as a manner of speaking about certain operator algebras. Our reasoning here is grounded in the basic properties of quantum sets, which thus serve as axioms. We refer to notions from operator algebra and category theory for context and for brevity; however, we seldom appeal directly to results from these subjects.

Finite-dimensional von Neumann algebras and trace-preserving completely positive maps appear in section~\ref{section 6}, where we describe the category $\cat{qGph}$ in terms that are more common to quantum information theory. Specifically we show that the morphisms of this category correspond to certain quantum channels that respect confusability structure in an appropriate sense \cite{zbMATH07388954}. These quantum channels are deterministic in the exotic sense that they never increase a measure of entropy that is related to but distinct from von Neumann entropy \cite{zbMATH07990859}.

\medskip

\noindent \textbf{Context.} The category $\cat{qGph}$ fits into prior literature in the following~way: Modulo the choice of formalization, an object of $\cat{qGph}$ is a quantum graph in the sense of Weaver \cite[Definition~2.6(d)]{zbMATH06008057} that need not be reflexive \cite[Definition~2.4(d)(i)]{zbMATH06008057} and that is on a hereditarily atomic von Neumann algebra \cite[Definition~5.3]{zbMATH07287276}. Therefore, the finite quantum graphs in $\cat{qGph}$ are the real, tracial, and self-adjoint (finite) quantum graphs in the terminology of \cite{zbMATH07668023}.

A morphism in $\cat{qGph}$ between two finite quantum graphs is a CP morphism \cite{zbMATH07388954} that is adjoint to a unital $\dagger$-homomorphism (\cref{R}, \cref{M}). We show that under the latter assumption, Weaver's two notions of a CP morphism coincide (\cref{O}). The stronger of these two notions generalizes Stahlke's notion of a homomorphism \cite[Definition~7]{zbMATH06709711}. Equivalently, a morphism in $\cat{qGph}$ between two finite quantum graphs is a homomorphism in the sense of Musto, Reutter, and Verdon \cite[Definition~5.4]{zbMATH06936038}; see \cite[Remark~5.9]{zbMATH06936038} and \cref{S}. These homomorphisms are related to the existence of local, i.e., classical strategies for generalizations of the graph homomorphism game, which we do not investigate in this article \cite{zbMATH07667940,arXiv:2408.15444}.

\pagebreak

\noindent \textbf{Notation.} 
When $H$ is a finite-dimensional Hilbert space, we write $L(H)$ for the space of linear operators on $H$, and when $a \in L(H)$, we write $a^\dagger$ for the Hermitian adjoint of $a$. We use this notation for the Hermitian adjoint throughout, speaking of $\dagger$-algebras and $\dagger$-homomorphisms, where others might speak of $*$-algebras and $*$-homomorphisms. When $\R$ and $\S$ are subspaces of $L(H)$, we write $\R, \S \leq L(H)$ and let $\R\cdot \S$ be the span of the set $\{rs \mid r \in \R, \, s \in \S\}$. Similarly, we write $\R^\dagger$ for the subspace $\{r^\dagger \mid r \in \R\}$. When $\cat{C}$ is a category, we write $\cat{C}(X,Y)$ for the set of all morphisms from an object $X$ to an object  $Y$.

\bigskip

\noindent \textbf{Acknowledgments.} We thank Matthew Daws for his comments. The first author was supported by the Air Force Office of Scientific Research under Award No.~FA9550-21-1-0041 and by the National Science Foundation under Award No.~DMS-2231414. The second author was supported by the Austrian Science Fund (FWF) under Project DOI 10.55776/PAT6443523.

\bigskip

\section{Quantum sets}\label{section 2}

This section provides a quick introduction to quantum sets \cite{zbMATH07287276}.

$C^*$-algebras are sometimes regarded as a quantum generalization of locally compact Hausdorff spaces because the Gelfand representation provides a duality between commutative $C^*$-algebras and locally compact Hausdorff spaces. This perspective is the bedrock of noncommutative geometry \cite[chapter~1]{zbMATH01452515}. In this conceptual framework, one imagines a category of quantum locally compact Hausdorff spaces that is dual to the category of $C^*$-algebras and Woronowicz morphisms. Formally, the former category is defined to be the opposite of the latter \cite{zbMATH03810280}.

Discrete spaces are a class of locally compact Hausdorff spaces, and quantum discrete spaces are commonly defined to correspond to $C^*$-algebras of the form $\A \iso \bigoplus_{i \in I}^{c_0} M_{n_i}(\CC)$, where the notation refers to the $c_0$-direct sum of $C^*$-algebras. This definition emerged from the theory of compact quantum groups \cite{zbMATH04152742}. Just as sets are naturally identified with discrete spaces, quantum sets may be naturally identified with discrete quantum spaces.

Formally, we define a \emph{quantum set} $\X$ to be a set of nonzero finite-dimensional Hilbert spaces, which we call the \emph{atoms} of $\X$. Intuitively, these Hilbert spaces are not the elements of $\X$, so we write $\X$ when we regard this object as a discrete quantum space and $\At(\X)$ when we regard this object as a set of Hilbert spaces. The operator algebras
$$c_0(\X) : = \bigoplus_{X \in \At(\X)}^{c_0} L(X), \qquad \qquad \ell^\infty(\X) : = \bigoplus_{X \in \At(\X)}^{\ell^\infty} L(X)$$
generalize the classical sequence spaces in the obvious way. As in the classical case, $c_0(\X)$ is always a $C^*$-algebra, and $\ell^\infty(\X)$ is always a von Neumann algebra.

It is generally preferable to work with the von Neumann algebras $\ell^\infty(\X)$ rather than the $C^*$-algebras $c_0(\X)$ for two reasons. First, functions $\X \to \Y$ correspond to Woronowicz morphisms $c_0(\Y) \to c_0(\X)$ and to unital normal $\dagger$-homomorphisms $\ell^\infty(\Y) \to \ell^\infty(\X)$, and the latter notion is technically simpler. Second, relations $\X \to \Y$ correspond to quantum relations \cite{zbMATH06008057} between $\ell^\infty(\X)$ and $\ell^\infty(\Y)$, and while we speculate that they correspond to a class of $C^*$-correspondences \cite{zbMATH01003154} between $c_0(\X)$ and $c_0(\Y)$ too, we are not aware of this connection in the literature. Up to isomorphism, a von Neumann algebra is of the form $\ell^\infty(\X)$ iff it is \emph{hereditarily atomic} \cite[Proposition~5.4]{zbMATH03810280}.

A \emph{relation} $R$ from $\X$ to $\Y$ is a choice of subspaces
$
R(X,Y) \leq L(X,Y)
$ for all $X \in \At(\X)$ and $Y \in \At(\Y)$. Quantum sets and relations form a dagger compact category $\cat{qRel}$ in a straightforward way \cite[Theorem.~3.6]{zbMATH03810280}, which may be guessed by the reader. A \emph{dagger compact category}, i.e., a strongly compact category \cite{AbramskyCoecke2004}, is a symmetric monoidal category that is further equipped with a choice of dual objects and a compatible contravariant functor $(-)^\dagger$. The category of $\cat{FinHilb}$ of finite-dimensional Hilbert spaces and linear operators and the category $\cat{Rel}$ of sets and relations are two prominent examples.

The categories $\cat{Rel}$ and $\cat{qRel}$ are similar in many respects. The former is the prototypical example of an \emph{allegory} \cite{zbMATH00045228}, and the latter behaves much like an allegory; both are \emph{dagger compact quantaloids} \cite{arXiv:2504.18266}. Hence, the morphisms of $\cat{Rel}$ and $\cat{qRel}$ are partially ordered, compatibly with the dagger compact structure.
In $\cat{Rel}$, $r \leq s$ if $r \subseteq s$, and in $\cat{qRel}$, $R \leq S$ if $R(X,Y) \leq S(X,Y)$ for all $X \in \At(\X)$ and $Y \in \At(\Y)$. A map in an allegory, such as $\cat{Rel}$, is defined to be a morphism $r\: X \to Y$ such that $r^\dagger \circ r \geq \id_X$ and $r \circ r^\dagger \leq \id_Y$ \cite[2.13]{zbMATH00045228}, and this definition is intelligible in $\cat{qRel}$ as well.

A \emph{function} $F\: \X \to \Y$ is a relation from $\X$ to $\Y$ such that $F^\dagger \circ F \geq \Id_\X$ and $F \circ F^\dagger \leq \Id_\Y$, where $\Id_\X$ and $\Id_\Y$ are identity relations. The resulting category $\cat{qSet}$ of quantum sets and functions is dual to the category of hereditarily atomic von Neumann algebras and unital normal $\dagger$-homomorphisms \cite[Theorem.~7.4]{zbMATH03810280}.

Similarly, the category $\cat{qRel}$ is equivalent to the category of hereditarily atomic von Neumann algebras and quantum relations. For any von Neumann algebras $\M$ and $\N$, a \emph{quantum relation} \cite{zbMATH06008057} from $\M \subseteq \B(H)$ to $\N \subseteq \B(K)$ is an ultraweakly closed subspace $\R \subseteq \B(H, K)$ such that $\N' \cdot \R \cdot \M' \subseteq \R$, where $\R \cdot \S$ is the ultraweakly closed span of operators of the form $rs$ for $r \in \R$ and $s \in \S$. In this incarnation of $\cat{qRel}$ \cite{arXiv:1101.1694}, each unital normal $\dagger$-homomorphism $\pi\: \N \to \M$ corresponds to the quantum relation
$$
\F = \{v \in \B(H, K) \mid b v = v \pi(b) \text{ for all } b \in \N\}.
$$

Both constructions of $\cat{qRel}$ have their virtues, and we appeal to both in this article. However, our approach is mostly agnostic with regard to the choice of construction, and  in this sense, we are working directly with quantum sets. We model $\X$ by the set $\At(\X)$ only in the proofs of Lemmas \ref{I} and \ref{W} and by the von Neumann algebra $\ell^\infty(\X)$ only in section \ref{section 6}.

\section{The symmetric monoidal category $\cat{qGph}$}\label{section 3}

In this section, we construct the symmetric monoidal category $\cat{qGph}$ within $\cat{qRel}$ and then prove that is closed. The same arguments suffice to construct the symmetric monoidal category $\cat{Gph}$ within $\cat{Rel}$ and to prove that it is closed, so this section may be viewed as a quantization of some elementary graph \emph{theory}.

\begin{definition}\label{A}
A \emph{quantum graph} is a pair $G = (\V_G, E_G)$ such that
\begin{enumerate}
\item $\V_G$ is a quantum set,
\item $E_G$ is a relation on $\V_G$ that satisfies $E_\X^\dagger = E_\X$.
\end{enumerate}
A \emph{homomorphism} $G \to H$ is a function $\Phi\: \V_G \to \V_H$ such that $$\Phi\circ E_G \leq E_H \circ \Phi.$$
The defines the category $\cat{qGph}$ of quantum graphs and their homomorphisms.
\end{definition}

It is routine to verify that $\cat{qGph}$ is a category with the obvious notion of composition. Its isomorphisms can be characterized in the following way.

\begin{lemma}\label{B}
Let $G$ and $H$ be quantum graphs, and let $\Phi:\V_G \to \V_H$ be a function. The following are equivalent:
\begin{enumerate}
\item $\Phi$ is an isomorphism $G \to H$ in $\cat{qGph}$,
\item $\Phi$ is a bijection, i.e., $\Phi^\dagger = \Phi^{-1}$, and $\Phi\circ E_G= E_H\circ \Phi$.
\end{enumerate}  
\end{lemma}

\begin{proof}
Assume that $\Phi\: G \to H$ is an isomorphism in $\cat{qGph}$, and let $\Psi\: H\to G$ be its inverse. We immediately have that $\Phi\circ E_G\leq E_H\circ \Phi$ and that $\Psi\circ E_H\leq E_G\circ \Psi$. Thus, $E_H\circ \Phi=\Phi\circ \Psi\circ E_H\circ \Phi\leq \Phi\circ E_G\circ \Psi\circ \Phi=\Phi\circ E_G$. We conclude that $\Phi\circ E_G=E_H\circ \Phi$. The function $\Phi$ is a bijection because $\Phi^\dagger = \Phi^\dagger \circ \Phi \circ \Psi \geq \Psi$ and $\Phi^\dagger = \Psi \circ \Phi \circ \Phi^\dagger \leq \Psi$.

Conversely, assume that $\Phi:\X\to \Y$ is a bijection such that $\Phi\circ E_G=E_H\circ \Phi$. We immediately have that $\Phi$ is a homomorphism. Furthermore, $\Phi^\dag:{\Y}\to \X$ is a function such that $\Phi^\dag\circ E_H=\Phi^\dag\circ E_H\circ \Phi\circ \Phi^\dag=\Phi^\dag\circ \Phi\circ E_G\circ \Phi^\dag=E_G\circ \Phi^\dag$, and hence, $\Phi^\dag$ is a homomorphism $(\Y,E_H)\to( \X,E_G)$. The functions $\Phi$ and $\Phi^\dagger$ are inverses because $\Phi$ is a bijection; we conclude that $\Phi$ is an isomorphism.
\end{proof}

We now use \cite[Definition~2.1]{arXiv:2408.11911} to make $\cat{qGph}$ a symmetric monoidal category. Recall that $\mathbf 1$ and $- \times -$ are the monoidal unit and product of $\cat{qRel}$, respectively.

\begin{definition}
We define a symmetric monoidal structure $\cat{qGph}$:
\begin{enumerate}
\item the monoidal unit is the quantum graph $K_1 = (\mathbf 1, 0_{\mathbf 1})$,
\item the monoidal product of quantum graphs $G$ and $H$ is
$$
G \boxprod H = ( \V_G \times \V_H, E_G \boxprod E_H),
$$
where $E_G \boxprod E_H = (E_G \times \Id_H) \vee (\Id_G \times E_H)$,
\item the monoidal product of homomorphisms $\Phi$ and $\Psi$ is simply $\Phi \times \Psi$.
\end{enumerate}
\end{definition}

To show that $G \boxprod H$ is a quantum graph, we reason that
\begin{align*}
(E_G & \boxprod  E_H)^\dag = ((E_G \times \Id_H) \vee (\Id_G \times E_H))^\dag = (E_G \times \Id_H)^\dag \vee (\Id_G \times E_H)^\dag \\ & = (E_G^\dag \times \Id_H) \vee (\Id_G \times E_H^\dag) = (E_G \times \Id_H) \vee (\Id_G \times E_H) = E_G \boxprod E_H.
\end{align*}
To show that $\Phi \times \Psi$ is a homomorphism, for homomorphisms $\Phi:G \to H$ and $\Psi: K \to L$, we reason that 
\begin{align*}
(\Phi\times \Psi)\circ (E_G\boxprod E_K) & = (\Phi\times \Psi)\circ \big((E_G\times \Id_K)\vee(\Id_G\times E_K)\big)\\
& = \big((\Phi\times \Psi)\circ (E_G\times  \Id_K)\big)\vee \big((\Phi\times \Psi)\circ (\Id_G\times E_K)\big)\\
& = \big((\Phi\circ E_G)\times(\Psi\circ \Id_K)\big) \vee \big((\Phi\circ \Id_G)\times(\Psi\circ E_K)\big)\\
& \leq \big((E_H\circ \Phi)\times (\Id_L\circ \Psi)\big)\vee \big((\Id_H\circ \Phi)\times(E_L\circ \Psi)\big)\\
& = \big((E_H\times \Id_L)\circ (\Phi\times \Psi)\big)\vee\big((\Id_H\times E_L)\circ (\Phi\times \Psi)\big)\\
& = \big((E_H\times \Id_L)\vee (\Id_H\times E_L)\big)\circ (\Phi\times \Psi)\\
& = (E_H\boxprod E_L)\circ (\Phi\times \Psi).
\end{align*}
It is routine to verify that the associator, braiding, and unitors of $\cat{qRel}$ are also the associator, braiding, and unitors for the monoidal unit and monoidal product of $\cat{qGph}$. For example, for each quantum graph $G$, the right unitor component $R:\V_G\times \mathbf 1\to \V_G$ satisfies
\begin{align*}&
R \circ (E_G \boxprod 0_{\mathbf 1})
=
R \circ ((E_G \times \Id_\one) \vee (\Id_\X \times 0_\one)) = R \circ (E_G \times \Id_\one) = E_G \circ R,
\end{align*}
by the naturality of the right unitor in $\cat{qRel}$.

The category $\cat{qSet}$ is symmetric monoidal closed \cite[section~9]{zbMATH07287276}. Thus, for all quantum sets $\X$ and $\Y$, we have a \emph{quantum function set} $\Y^\X$ and a function $\Eval_{\X,\Y}\: \Y^\X \times \X \to \Y$ that is universal among such functions. Explicitly, for all functions $F\: \W \times \X \to \Y$, there exists a unique function $\W \to \Y^\X$ that makes the following diagram commute.
\[
\begin{tikzcd}
\W \times \X
\arrow{rrrd}{F}
\arrow[dotted]{d}[swap]{! \times \Id_\X}
&
&
&
\\
\Y^\X \times \X
\arrow{rrr}[swap]{\Eval_{\X,\Y}}
&
&
&
\Y
\end{tikzcd}
\]

We now show that the symmetric monoidal category $\cat{qGph}$ is closed as well. The idea is to construct a quantum graph of homomorphisms $G \to H$, whose vertex quantum set is a subset of the quantum function set $(\V_H)^{\V_G}$. To do so, we briefly recall that a \emph{subset} of a quantum set $\X$ is a simply a subobject of $\X$ in $\cat{qSet}$, i.e., a quantum set $\W$ and a monomorphism $J_\W\: \W \to \X$, modulo the obvious notion of equivalence. Concretely, the subsets of $\X$ are the quantum sets $\W$ with $\At(\W) \subseteq \At(\X)$, and $J_\W$ is then defined by $J_\W(W,W) = \mathrm{span}(1_W)$.

\begin{definition}\label{C}
Let $G$ and $H$ be quantum graphs. We define the quantum graph $[G,H]$:
\begin{enumerate}[ref = \thelemma(\arabic*)]
\item $\V_{[G,H]}$ is the largest subset $\W$ of the quantum function set $(\V_H)^{\V_G}$ such that $ \Eval_{G,H}\circ (J_{\W} \times \Id_G)$ is a homomorphism $(\W \times \V_G, \Id_\W \times E_G) \to H$.\label[definition]{C1}
\item $E_{[G,H]}$ is the largest symmetric relation $Q$ on $\W = \V_{[G,H]}$ such that $\Eval_{G,H}\circ (J_{\W} \times \Id_G)$ is a homomorphism $(\W \times \V_G, Q \times \Id_G) \to H$.\label[definition]{C2}
\end{enumerate}
\end{definition}

It is not immediately obvious that there exist a largest subset $\W$ and a largest symmetric relation $Q$ in \cref{C}. The existence of a largest subset $\W$ such that $\Eval_{G,H}\circ (J_{\W} \times \Id_G)$ is a homomorphism follows by \cite[Lemma 8.2]{zbMATH07605379}, where the idea is that the set of all such subsets $\W$ is closed under unions. Similarly, the existence of a largest relation $Q$ such that $\Eval_{G,H}\circ (J_{\W} \times \Id_G)$ is a homomorphism follows as in the proof of the same \cite[Lemma 8.2]{zbMATH07605379}, where the idea is that the set of all such relations $Q$ is closed under joins. This largest relation $Q$ is symmetric because, writing $\Phi = \Eval_{G,H} \circ (J_\W \times \Id_G)$, we calculate that
\begin{align*}
\Phi \circ (Q^\dagger \times \Id_G) & \leq \Phi \circ (Q^\dagger \times \Id_G) \circ \Phi^\dagger \circ \Phi = (\Phi \circ (Q \times \Id_G) \circ \Phi^\dagger)^\dagger \circ \Phi \\ &  \leq (E_H \circ \Phi \circ \Phi^\dagger)^\dagger \circ \Phi
\leq E_H^\dagger \circ \Phi = E_H \circ \Phi,
\end{align*}
which implies that $Q^\dagger \leq Q$ and, hence, that $Q^\dagger = Q$.

\begin{theorem}\label{D}
Let $G$ and $H$ be quantum graphs, and let $$\Eval_{[G,H]} = \Eval_{G,H}\circ (J_{[G,H]} \times \Id_G).$$
The function $\Eval_{[G,H]}$ is a homomorphism $[G,H] \boxprod G \to H$. Furthermore, for all quantum graphs $K$ and all homomorphisms, $\Phi\: K \boxprod G \to H$, there exists a unique homomorphism $\Psi\: K \to [G, H]$ such that $\Eval_{[G,H]} \circ (\Psi \times \Id_G) = \Phi$.
\[
\begin{tikzcd}
K \boxprod G
\arrow{rrrd}{\Phi}
\arrow[dotted]{d}[swap]{\Psi \times \Id_G}
&
&
&
\\
{[G,H]} \boxprod G
\arrow{rrr}[swap]{\Eval_{[G,H]}}
&
&
&
H
\end{tikzcd}
\]
Thus, the symmetric monoidal category $\cat{qGph}$ is closed \cite[section~VII.7]{zbMATH01216133}.
\end{theorem}

\begin{proof}
By \cref{C1}, we have that $$\Eval_{[G,H]} \circ (\Id_{[G,H]} \times E_G) \leq E_H \circ \Eval_{[G,H]},$$
and by \cref{C2}, we have that
$$
\Eval_{[G,H]} \circ (E_{[G,H]} \times \Id_G) \leq E_H \circ \Eval_{[G,H]},
$$
so overall we have that $\Eval_{[G,H]} \circ (E_{[G,H]} \boxprod E_G) \leq E_H \circ \Eval_{[G,H]}$. Therefore, $\Eval_{[G,H]}$ is a homomorphism $[G,H] \boxprod G \to H$.

Let $K$ be a quantum graph, and let $\Phi\: K \boxprod G \to H$ be a homomorphism. In particular, $\Phi$ is a function $\V_K \times \V_G \to \V_H$, and since $\cat{qSet}$ is symmetric monoidal closed, there exists a unique function $F\: \V_K \to (\V_H)^{\V_G}$ such that $\Eval_{G,H} \circ (F \times \Id_G) = \Phi$.
\[
\begin{tikzcd}
\V_K \times \V_G
\arrow{rrrd}{\Phi}
\arrow[dotted]{d}[swap]{F \times \Id_G}
&
&
&
\\
(\V_H)^{\V_G} \times \V_G
\arrow{rrr}[swap]{\Eval_{G, H}}
&
&
&
\V_H
\end{tikzcd}
\]
We factor $F$ through a subset $\X$ of $(\V_H)^{\V_G},$
\[
\begin{tikzcd}
\V_K
\arrow{rr}{\overline F}
\arrow[bend right = 10]{rrrr}[swap]{F}
&
&
\X
\arrow{rr}{J_\X}
&
&
(\V_H)^{\V_G}
\end{tikzcd}
\]
such that $\overline F$ is a surjection, i.e., such that $\overline F \circ \overline F^\dagger = \Id_\X$ \cite[Definition~3.2]{zbMATH07605379}.

We now calculate that
\begin{align*}
\Eval_{G,H} \circ {} & (J_\X\times \Id_G) \circ (\Id_\X \times E_G)
\\ & =
\Eval_{G,H} \circ (J_\X\times \Id_G) \circ  (\overline F\times \Id_G) \circ (\overline F^\dagger \times \Id_G) \circ  (\Id_\X \times E_G)
\\ &=
\Eval_{G,H} \circ  ( F \times \Id_G)  \circ  (\Id_K \times E_G)\circ (\overline F^\dagger \times \Id_G)
\\ &=
\Phi \circ  (\Id_K \times E_G)\circ (\overline F^\dagger \times \Id_G)
\\ &\leq
\Phi \circ  (E_K \boxprod E_G)\circ (\overline F^\dagger \times \Id_G)
\\ &\leq
E_H \circ \Phi \circ (\overline F^\dagger \times \Id_G)
\\ &=
E_H \circ \Eval_{G,H} \circ  ( F \times \Id_G) \circ (\overline F^\dagger \times \Id_G)
\\ & =
E_H \circ \Eval_{G,H} \circ (J_\X\times \Id_G) \circ  (\overline F \times \Id_G) \circ (\overline F^\dagger \times \Id_G)
\\ &=
E_H \circ \Eval_{G,H} \circ (J_\X\times \Id_G),
\end{align*}
where the first inequality follows by definition of the box product and the second inequality from the fact that $\Phi$ is a homomorphism.
Thus, $\Eval_{G,H} \circ (J_\X \times \Id_G)$ is a homomorphism from $(\X \times \V_G, \Id_\X \times E_G)$ to $H$. The maximality of $\V_{[G,H]}$ in \cref{C1} now implies that $\X$ is a subset of $\V_{[G,H]}$. Let $\Psi = J_\X^{[G,H]} \circ \overline F$, where $J_\X^{[G,H]} \: \X \to \V_{[G,H]}$ is the inclusion function, and note that $J_{[G,H]} \circ \Psi = J_{[G,H]} \circ J_\X^{[G,H]}  \circ \overline F = J_\X \circ \overline F = F$.

We show that $\Psi$ is a homomorphism $K \to [G,H]$ by combining two inequalities. First, we calculate that
\begin{align*}
    \Eval_{G,H}\circ {} & (J_{[G,H]}\times \Id_G)\circ((\Psi\circ E_K\circ \Psi^\dag)\times \Id_G)\\
    & =   \Eval_{G,H}\circ ((J_{[G,H]}\circ \Psi\circ E_K\circ \Psi^\dag)\times \Id_G) \\
    & =    \Eval_{G,H}\circ ((F\circ E_K\circ \Psi^\dag)\times \Id_G)\\
    & =  \Eval_{G,H}\circ (F\times \Id_G)\circ (E_K\times \Id_G)\circ (\Psi^\dag\times \Id_G)\\
    & =  \Phi\circ (E_K\times \Id_\X)\circ (\Psi^\dag\times \Id_G)\\
    & \leq \Phi\circ (E_K\square E_G)\circ (\Psi^\dag\times \Id_G)\\
    & \leq E_H\circ \Phi\circ (\Psi^\dag\times \Id_G)\\
    & = E_H\circ \Eval_{G,H}\circ (F\times \Id_G)\circ (\Psi^\dag\times \Id_G)\\
    & = E_H\circ \Eval_{G,H}\circ ((J_{[G,H]}\circ \Psi\circ \Psi^\dag)\times \Id_G)\\
    & \leq E_H\circ \Eval_{G,H}\circ (J_{[G,H]}\times \Id_G),
\end{align*}
where we have the first inequality by the definition of the box product, the second inequality by the definition of a homomorphism, and the last inequality by the definition of a function. Second, we observe that 
\begin{align*} &
    \Eval_{G,H}\circ (J_{[G,H]}\times \Id_G)\circ (\Id_{[G,H]}\times E_G)
    \leq E_H\circ \Eval_{G,H}\circ (J_{[G,H]}\times \Id_G)
\end{align*}
by \cref{C1}, which is the definition of $\V_{[G,H]}$.
Since the composition of relations  preserves joins, we find that
\begin{align*}
    \Eval_{G,H} {} & \circ  (J_{[G,H]}\times \Id_G) \circ\big((\Psi\circ E_K\circ \Psi^\dag)\square E_G\big)\\
    & =    \Eval_{G,H}\circ (J_{[G,H]}\times \Id_G)\circ\big(((\Psi\circ E_K\circ \Psi^\dag)\times \Id_G)\vee (\Id_{[G,H]}\times E_G)\big)\\
     & = \big(\Eval_{G,H}\circ ((J_{[G,H]}\circ \Psi\circ E_K\circ \Psi^\dag)\times \Id_G)\big) \\ & \qquad \vee\big(\Eval_{G,H}\circ (J_{[G,H]}\times \Id_G)\circ(\Id_{[G,H]}\times E_G)\big)\\
    & \leq E_H\circ \Eval_{G,H}\circ (J_{[G,H]}\times \Id_G).
\end{align*}
Thus, $\Eval_{G,H} \circ  (J_{[G,H]}\times \Id_G)$ is a homomorphism from $(\V_{[G,H]}, \Psi \circ E_K \circ \Psi^\dagger)\boxprod G$ to $H$. It follows that \[\Psi \circ E_K \circ \Psi^\dagger \leq E_{[G,H]}\]
by \cref{C2}, which is the definition of $E_{[G,H]}$. We conclude that
\[
\Psi \circ E_K  = \Psi \circ E_K \circ \Id_K \leq \Psi \circ E_K \circ \Psi^\dagger \circ \Psi \leq E_{[G,H]} \circ \Psi.
\]
Therefore, $\Psi$ is a homomorphism $K \to [G, H]$.

We certainly have that $\Eval_{[G,H]} \circ (\Psi \times \Id_G) = \Phi$ because
\begin{align*}
\Eval_{[G,H]} \circ {} & (\Psi \times \Id_G) = 
\Eval_{G,H} \circ (J_{[G,H]} \times \Id_G) \circ (\Psi \times \Id_G) \\ & = \Eval_{G,H} \circ ((J_{[G,H]} \circ \Psi) \times \Id_G) \\ & = \Eval_{G,H} \circ (F \times \Id_G) = \Phi.
\end{align*}
Furthermore, $\Psi$ is the unique homomorphism $K \to [G,H]$ with this property because any such homomorphism $\Psi'$ we may reason that
\begin{align*}
\Eval_{[G,H]} \circ (\Psi \times \Id_G) & = \Eval_{[G,H]} \circ (\Psi' \times \Id_G), \\
\Eval_{G,H} \circ (J_{[G,H]} \times \Id_G) \circ (\Psi \times \Id_G) & = \Eval_{G,H} \circ (J_{[G,H]} \times \Id_G) \circ (\Psi' \times \Id_G) \\
\Eval_{G,H} \circ ((J_{[G,H]} \circ \Psi) \times \Id_G) & = \Eval_{G,H} \circ ((J_{[G,H]} \circ \Psi') \times \Id_G),\\
J_{[G,H]} \circ \Psi & = J_{[G,H]} \circ \Psi', \\
\Psi & = \Psi'.
\end{align*}
We have used the universal property of $\Eval_{G,H}$ and that $J_{[G,H]}$ is monic in $\cat{qSet}$. Therefore, $\Psi$ is the unique homomorphism such that $\Eval_{[G,H]} \circ (\Psi \times \Id_G) = \Phi$, and more generally, $\cat{qGph}$ is monoidal closed.
\end{proof}

The conclusion that a symmetric monoidal category is closed has some well-known consequences, e.g. \cite[Proposition~6.1.7]{zbMATH00692330}, which we record as a corollary.

\begin{corollary}
The construction in \cref{C} extends  to a functor
$$
[-,-]\: \cat{qGph}^{op} \times \cat{qGph} \to \cat{qGph}
$$
that maps colimits to limits in the first argument and preserves limits in the second argument. Similarly, the functor
$$
- \boxprod - \: \cat{qGph} \times \cat{qGph} \to \cat{qGph}
$$
preserves colimits in each argument.
\end{corollary}

It is routine to verify that we can construct a coproduct $G + H$ of quantum graphs $G$ and $H$ by $\V_{G +H} = \V_G + \V_H$ and $E_{G+H} = E_G + E_H$, where $\X + \Y$ denotes the coproduct of $\X$ and $\Y$ in $\cat{qSet}$ \cite[Remark~3.7]{zbMATH07287276}. Thus, we infer that
$$
G \boxprod (H_1 + H_2) \iso G \boxprod H_1 + G \boxprod H_2, \quad (G_1 + G_2) \boxprod H \iso G_1 \boxprod H + G_2 \boxprod H.
$$

\begin{remark}\label{E}
The category $\cat{Gph}$ of graphs and homomorphisms is similarly a closed symmetric monoidal category with the box product. To some extent, this is folklore, and in any case, this is routine to verify. However, we remark, that the arguments in this section are valid in $\cat{Rel}$ as in $\cat{qRel}$. Thus, it also serves to establish that $\cat{Gph}$ is a closed symmetric monoidal category.
\end{remark}

\section{The enrichment of $\cat{qGph}$ over $\cat{Gph}$}\label{section 4}

Adjunctions can be characterized in several ways. In particular, \cref{D} implies that there is a natural bijection
$$
\cat{qGph}(K \boxprod G, H) \iso \cat{qGph}(K, [G, H]).
$$
In this section, we observe that both sides of this bijection carry a canonical graph structure and that it is a natural isomorphism of graphs.

The symmetric monoidal category $\cat{qGph}$ is canonically \emph{enriched} over the symmetric monoidal category $\cat{Gph}$. This means that the hom sets $\cat{qGph}(G, H)$ can all be made into graphs in a way that is compatible with all the structure of $\cat{qGph}$ as a symmetric monoidal category.

Explicitly, we define two homomorphisms $\Phi_1, \Phi_2 \: G \to H$ to be adjacent if $\Phi_1^\dagger \circ E_{H} \circ \Phi_2 \geq \Id_G$. This definition makes the composition a homomorphism
$$
\circ\: \cat{qGph}(G, G') \boxprod \cat{qGph}(G',K) \to \cat{qGph}(G, G'')
$$
because $\Phi_1 \sim \Phi_2$ implies $\Psi \circ \Phi_1 \sim \Psi \circ \Phi_2$ and $\Psi_1 \sim \Psi_2$ implies $\Psi_1 \circ \Phi \sim \Psi_2 \circ \Phi$. Indeed, for all $\Phi_1, \Phi_2 \: G \to G'$ and all $\Psi\: G' \to G''$, if $ \Phi_1^\dagger \circ E_{G'} \circ \Phi_2 \geq \Id_G$, then
\begin{align*}
(\Psi \circ \Phi_1)^\dagger \circ E_{G''} \circ (\Psi \circ \Phi_2)
& =
\Phi_1^\dagger \circ \Psi^\dagger \circ E_{G''} \circ \Psi \circ \Phi_2
\geq
\Phi_1^\dagger \circ \Psi^\dagger \circ \Psi \circ E_{G'} \circ \Phi_2
\\ & \geq
\Phi_1^\dagger \circ \Id_{G'} \circ E_{G'} \circ \Phi_2
=
\Phi_1^\dagger \circ E_{G'} \circ \Phi_2
\geq
\Id_G.
\end{align*}
The proof that $\Psi_1 \sim \Psi_2$ implies $\Psi_1 \circ \Phi \sim \Psi_2 \circ \Phi$ is entirely similar. Furthermore, the function $\{\ast\} \to \cat{qGph}(G,G)$ whose value is the identity homomorphism is a automatically a homomorphism $K_1 \to \cat{qGph}(G,G)$, where $K_1$ is the graph with one vertex and no edges. Thus, $\cat{qGph}$ is enriched over $\cat{Gph}$ as a category.

The verification that $\cat{qGph}$ is enriched over $\cat{Gph}$ as a symmetric monoidal category amounts to more of the same. We simply observe that the monoidal product of morphisms is a homomorphism
$$
\boxprod\: \cat{qGph}(G, G') \boxprod \cat{qGph}(H, H') \to \cat{qGph} (G \boxprod G', H \boxprod H')
$$
because $\Phi_1 \sim \Phi_2$ implies $\Phi_1 \times \Psi \sim \Phi_2 \times \Psi$ and $\Psi_1 \times \Psi_2$ implies $\Phi \times \Psi_1 \sim \Phi \times \Psi_2$. Indeed, for all $\Phi_1, \Phi_2 \: G \to G'$ and all $\Psi\: H \to H'$, if $ \Phi_1^\dagger \circ E_{G'} \circ \Phi_2 \geq \Id_G$, then
\begin{align*}
(\Phi_1 \times \Psi)^\dagger  &\circ (E_{G'} \boxprod E_{H'}) \circ (\Phi_2 \times \Psi)
\\ & =
(\Phi_1 \times \Psi)^\dagger \circ ((E_{G'} \times \Id_{H'})\vee (\Id_{G'} \times E_{H'})) \circ (\Phi_2 \times \Psi)
\\ & \geq 
(\Phi_1 \times \Psi)^\dagger \circ (E_{G'} \times \Id_{H'}) \circ (\Phi_2 \times \Psi)
\\ & =
(\Phi_1^\dagger \circ E_{G'} \circ \Phi_2) \times (\Psi^\dagger \circ \Id_{H'} \circ \Psi)
\geq
\Id_G \times \Id_H.
\end{align*}
The associators, braidings, and unitors do correspond to homomorphisms out of $K_1$, as in the case of identity morphisms. The coherence diagrams commute in the enriched sense because they commute in usual sense.

\begin{proposition}\label{F}
The bijection $\cat{qGph}(K, [G, H]) \to \cat{qGph}(K \boxprod G, H)$ that is defined by
$$
\Psi \mapsto \Eval_{[G,H]} \circ (\Psi \times \Id_G)
$$
is an isomorphism of graphs.
\end{proposition}

\begin{proof}
Let $\Psi_1, \Psi_2\: K \to [G,H]$ be homomorphisms, and let $\Phi_i =  \Eval_{[G,H]} \circ (\Psi_i \times \Id_G)$ for each $i \in \{1,2\}$. If $\Psi_1 \sim \Psi_2$, then $\Psi_1 \times \Id_G \sim \Psi_2 \times \Id_G$, and hence, $\Eval_{[G,H]} \circ (\Psi_1 \times \Id_G) \sim \Eval_{[G,H]} \circ (\Psi_2 \times \Id_G)$, all because $\cat{qGph}$ is enriched over $\cat{Gph}$ as a symmetric monoidal category. Therefore, $\Psi_1 \sim \Psi_2$ implies $\Phi_1 \sim \Phi_2$.

Assume that $\Phi_1 \sim \Phi_2$. By definition, this means that
$$
(\Psi_1^\dagger \times \Id_G) \circ \Eval_{[G,H]}^\dagger \circ E_H \circ \Eval_{[G,H]} \circ (\Psi_2 \times \Id_G) \geq \Id_K \times \Id_G
$$
and hence implies that
$$
\Eval_{[G,H]}^\dagger \circ E_H \circ \Eval_{[G,H]} \geq (\Psi_1 \circ \Psi_2^\dagger) \times \Id_G.
$$
Similarly,
$$
\Eval_{[G,H]}^\dagger \circ E_H \circ \Eval_{[G,H]} \geq (\Psi_2 \circ \Psi_1^\dagger) \times \Id_G,
$$
and thus,
$$
\Eval_{[G,H]}^\dagger \circ E_H \circ \Eval_{[G,H]} \geq ((\Psi_1 \circ \Psi_2^\dagger)\vee(\Psi_2 \circ \Psi_1^\dagger)) \times \Id_G.
$$
The relation $(\Psi_1 \circ \Psi_2^\dagger)\vee(\Psi_2 \circ \Psi_1^\dagger)$ is symmetric, so
this inequality expresses that the function $\Eval_{[G,H]}\: \V_{[G,H]} \times \V_G \to \V_H$ is a homomorphism for the relations $((\Psi_1 \circ \Psi_2^\dagger)\vee(\Psi_2 \circ \Psi_1^\dagger))\times \Id_G$ and $E_H$ in the sense of \cref{A}.

By \cref{C2}, we find that $E_{[G,H]} \geq (\Psi_1 \circ \Psi_2^\dagger)$ or, equivalently, that $\Psi_1^\dagger \circ E_{[G,H]} \circ \Psi_2 \geq \Id_K$. Therefore, $\Psi_1 \sim \Psi_2$, establishing that $\Phi_1 \sim \Phi_2$ implies $\Psi_1 \sim \Psi_2$. We conclude that this bijection is an isomorphism of graphs.
\end{proof}

Thus, the symmetric monoidal category $\cat{qGph}$ is closed in a $\cat{Gph}$-enriched sense. Since the symmetric monoidal category $\cat{Gph}$ is closed by Remark~\ref{E}, it too is closed in a $\cat{Gph}$-enriched sense. We conclude this section by discussing the relationship between these two categories, which consists of a pair of symmetric monoidal functors, $\Inc\: \cat{Gph} \to \cat{qGph}$ and $\cat{qGph}(K_1, -)\: \cat{qGph} \to \cat{Gph}$. 

The functor $\mathrm{Inc}$ is defined via the canonical inclusion functor $\cat{Rel} \to \cat{qRel}$, which is a full and faithful symmetric monoidal functor \cite[section~3]{zbMATH07287276}. This canonical inclusion functor maps each set $X$ to the quantum set $`X$ whose atoms are one-dimensional and in one-to-one correspondence with the elements of $X$; we write $\At(`X) = \{\CC_x\mid x \in X \}$. It maps each relation $r$ from $X$ to $Y$ to the relation $`r$ from $`X$ to $`Y$ such that $`r(\CC_x, \CC_y) = L(\CC_x, \CC_y)$ iff $(x,y) \in r$.

\begin{definition}\label{G}
We define the symmetric monoidal functors $\Inc$ and $\cat{qGph}(K_1, -)$.
\begin{enumerate}
\item The functor 
$\Inc\: \cat{Gph} \to \cat{qGph}$ is defined by
\begin{enumerate}
\item $\Inc(V_G, \sim_G) = (`V_G, `\sim_G)$ for each graph $G = (V_G, \sim_G)$,
\item $\Inc(\phi) = `\phi$ for each homomorphism $\phi\: G \to H$.
\end{enumerate}
\item The functor
$
\cat{qGph}(K_1, -)\: \cat{qGph} \to \cat{Gph}
$
is defined by
\begin{enumerate}
\item $\cat{qGph}(K_1, G)$ is a graph because $\cat{qGph}$ is enriched over $\cat{Gph}$,
\item $\cat{qGph}(K_1, \Phi)$ is postcomposition by $\Phi$.
\end{enumerate}
\end{enumerate}
\end{definition}

It is routine to verify that both $\Inc$ and $\cat{qGph}(K_1, -)$ are strong symmetric monoidal functors. More care is required to prove that they are adjoint.

\begin{theorem}[also \cref{X}]\label{H}
The functor $\Inc\:\cat{Gph} \to \cat{qGph}$ is left adjoint to the functor $\cat{qGph}(K_1, -) \: \cat{qGph} \to \cat{Gph}$. This adjunction is enriched over $\cat{Gph}$ in the sense that the natural bijection
$$
\cat{Gph}(G, \cat{qGph}(K_1,H)) \iso \cat{qGph}(\Inc(G), H),
$$
where $G$ is a graph and $H$ is a quantum graph, is an isomorphism.
\end{theorem}

This plausible statement is proved in appendix~\ref{appendix A}. The functor $\Inc$ is full and faithful because the canonical inclusion functor $\cat{Rel} \to \cat{qRel}$ is full and faithful. It formalizes the intuition that every graph is a quantum graph. We will often suppress $\Inc$, writing $G$ as an abbreviation for $\Inc(G)$. The functor $\cat{qGph}(K_1, -)$ formalizes the intuition that every quantum graph has a maximal classical subgraph, as it happens with other quantum structures \cite{zbMATH06566391}.

\section{Graph homomorphism games}\label{section 5}

In this section, we prove that the quantum graphs $[G,H]$ of \cref{C} encode the existence of winning quantum strategies for graph homomorphism games in a simple way: the $(G,H)$-homomorphism game has a winning quantum strategy iff the quantum graph $[G,H]$ is nonempty.

\begin{definition}
Let $G$ and $H$ be finite graphs.
The \emph{$(G,H)$-homomorphism game} is played in a single round by three players: Alice, Bob, and a verifier. Play proceeds in the following way:
\begin{enumerate}
\item the verifier selects a pair of vertices $g_1, g_2 \in V_G$;
\item Alice selects a vertex $h_1 \in V_H$ and Bob selects a vertex $h_2 \in V_H$;
\item the victory condition for Alice and Bob is
$$
(g_1 = g_2 \mathrel{\Rightarrow} h_1 = h_2)\mathrel{\&} (g_1 \sim g_2 \mathrel{\Rightarrow} h_1 \sim h_2).
$$
\end{enumerate}
\end{definition}

This victory condition may be regarded as a function $$w\:V_G \times V_H \times V_H \times V_G \to \{0,1\}$$
that maps each play $(g_1, h_1, h_2, g_2)$ to $1$ if Alice and Bob win and to $0$ if they lose. In extreme generality, a strategy for Alice and Bob is a morphism $$s\:V_G \times V_G \to V_H \times V_H,$$
in some category that contains $\cat{Set}$ as a subcategory. It is a winning strategy if $w \circ (\id \times s \times \id)$ is constantly $1$ on the Cartesian square of the diagonal:
\[
\begin{tikzcd}[column sep = 12ex]
V_G \times V_G
\arrow{rrd}[swap]{\mathrm{cnst_1}}
\arrow{r}{\Delta_G \times \Delta_G}
&
V_G \times V_G \times V_G \times V_G
\arrow{r}{\id_G \times s \times \id_G}
&
V_G \times V_H \times V_H \times V_G
\arrow{d}{w}
\\
&
&
\{0,1\}
\end{tikzcd}
\]

Of course, in this extreme generality, a winning strategy exists in $\cat{Set}$ itself whenever $H$ has a pair of adjacent vertices. The existence of a winning strategy is only interesting under some additional constraint that expresses that Alice must select $h_1$ without knowing $g_2$ and, similarly, Bob must select $h_2$ without knowing $g_1$.

\begin{example}
The category $\cat{Stoch}$ is defined as follows:
\begin{enumerate}
\item an object of $\cat{Stoch}$ is a set;
\item a morphism $X \to Y$ of $\cat{Stoch}$ is a \emph{stochastic map} from $X$ to $Y$, i.e, a function from $X$ to the set of all discrete probability distributions on $Y$.
\end{enumerate}
The category $\cat{Set}$ may be regarded as a subcategory of $\cat{Stoch}$ by identifying the elements of each set $Y$ with the deterministic probability distributions on $Y$.

A \emph{classical strategy} for Alice and Bob is a morphism $s$ in the category $\cat{Stoch}$ of sets and stochastic maps that is factored as
$$
\begin{tikzcd}[column sep = 12ex]
V_G \times V_G
\arrow{r}{\id_G \times m \times \id_G}
\arrow[bend right = 10]{rr}[swap]{s}
&
V_G \times X_1 \times X_2 \times V_G
\arrow{r}{s_1 \times s_2}
&
V_H \times V_H,
\end{tikzcd}
$$
where $m\: \{\ast\} \to X_1 \times X_2$, $s_1\: V_G \times X_1 \to V_H$, and $s_2\: X_2 \times V_G \to V_H$ for some sets $X_1$ and $X_2$. Intuitively, $s_1$ is the strategy that is played by Alice, $s_2$ is the strategy that is played by Bob, and $m$ is their shared randomness, e.g., a one-time pad. It is a standard exercise to show that Alice and Bob have a classical winning strategy iff there is a homomorphism from $G$ to $H$.

If $H$ is a complete simple graph, then Alice and Bob have a classical winning strategy iff there is a coloring of $G$ by the vertices in $H$, because a graph homomorphism from $G$ to $H$ is the same thing as a coloring of $G$ by the vertices of $H$. Alice and Bob can use this strategy as a zero-knowledge proof that they possess such a graph coloring. They can use their graph coloring to consistently win the $(G,H)$-homomorphism game against the verifier, and they can use their shared randomness to conceal their graph coloring by randomly permuting the vertices of $H$ between games.
\end{example}

\begin{example}
The category $\cat{qStoch}$ is defined as follows:
\begin{enumerate}
\item an object of $\cat{qStoch}$ is a quantum set;
\item a morphism $\X \to \Y$ of $\cat{qStoch}$ is a \emph{stochastic map} from $\X$ to $\Y$, i.e., a trace-preserving completely positive map from $\ell^1(\X)$ to $\ell^1(\Y)$, 
\end{enumerate}
where $\ell^1(\X) = \bigoplus_{X \in \At(\X)}^{\ell^1} L(X)$. Thus, a stochastic map $\X \to \Y$ is a stochastic map $X \to Y$ when $\X = `X$ and $\Y = `Y$ and a trace-preserving completely positive map $M_m(\CC) \to M_n(\CC)$ when $\CC^m$ and $\CC^n$ are the unique atoms of $\X$ and $\Y$, respectively.

The category $\cat{Stoch}$ may be regarded as a subcategory of $\cat{qStoch}$. Here, the ``inclusion'' functor maps each set $X$ to the quantum set $`X$ and each stochastic map $f\: X \to Y$ to the linear map $\ell^1(X) \to \ell^1(Y)$ that is defined by
$$
a \mapsto \sum_{x \in X} a(x) f(x).
$$
It follows that $\cat{Set}$ may also be regarded as a subcategory of $\cat{qStoch}$.

A \emph{quantum strategy} for Alice and Bob is a morphism $s$ in the category $\cat{qStoch}$ of quantum sets and stochastic maps that is factored as
$$
\begin{tikzcd}[column sep = 12ex]
V_G \times V_G
\arrow{r}{\id_G \times m \times \id_G}
\arrow[bend right = 10]{rr}[swap]{s}
&
V_G \times \X_1 \times \X_2 \times V_G
\arrow{r}{s_1 \times s_2}
&
V_H \times V_H,
\end{tikzcd}
$$
where $m\: \{\ast\} \to \X_1 \times \X_2$, $s_1\: V_G \times \X_1 \to V_H$, and $s_2\: \X_2 \times V_G \to V_H$ for some quantum sets $\X_1$ and $\X_2$. Intuitively, $s_1$ is the strategy that is played by Alice, $s_2$ is the strategy that is played by Bob, and $m$ is their shared entanglement, e.g., a Bell state. It has been shown that a winning quantum strategy may exist even when there is no homomorphism from $G$ to $H$ \cite{zbMATH06555288}.

Man\v{c}inska and Roberson obtained a criterion for the existence of winning quantum strategies. In effect, they defined a quantum strategy to use \emph{atomic} quantum sets $\X$ and $\Y$ and a \emph{pure} state $m$, but this definition is equivalent to the definition that we have given here by the same argument as for classical strategies.
Their criterion is the following: there is a winning quantum strategy for the $(G,H)$-homomorphism game iff there is a positive integer $n$ and a family of $n \times n$ complex projection matrices $\{p_{gh}: g \in V_G, h \in V_H\}$ such that
\begin{enumerate}
\item $\sum_{h \in V_H} p_{gh}$ is the identity matrix for all $g \in V_G$,
\item $p_{g_1h_1}p_{g_2h_2} = 0_n$ whenever $w(g_1, h_1, h_2, g_2) = 0$.
\end{enumerate}
\end{example}

Man\v{c}inska and Roberson reformulated their criterion in terms of the graph $M(G,n)$, which we now define. We will then use this graph to relate the $(G,H)$-homomorphism game to the quantum homomorphism graph $[G,H]$.

\begin{definition}[{\cite[section~2.4]{zbMATH06555288}}]
Let $G$ be a finite graph, and let $n$ be a positive integer. We define $M(G,n)$ to be the following graph:
\begin{enumerate}
\item a vertex is a function $p \: V_G \to M_n(\CC)$ such that $\sum_{g \in V_G} p(g) = 1_n$ with each matrix $p(g)$ being a projection,
\item two vertices $p_1$ and $p_2$ are adjacent if $p_1(g_1)p_2(g_2) = 0_n$ for all $g_1 \not \sim g_2$.
\end{enumerate}
\end{definition}

\begin{lemma}\label{I}
Let $G$ be a finite graph, and let $n\geq 1$. Then, $$M(G,n) \iso \cat{qGph}(Q_n, G),$$
where $\Q_n$ is the quantum  set with $\At(\Q_n) = \CC^n$ and $Q_n = (\Q_n, 0_{\Q_n})$.
\end{lemma}

\begin{proof}
We obtain this isomorphism by composing bijections $$V_{M(G,n)} \iso \mathrm{Hom}(\ell^\infty(V_G), M_n(\CC)) \iso  \cat{qSet}(\Q_n, `V_G),$$
where the expression $\mathrm{Hom}(\ell^\infty(V_G), M_n(\CC))$ denotes the set of all unital normal $\dagger$-homomorphisms $\ell^\infty(V_G) \to M_n(\CC)$. The first bijection maps a vertex $p$ of $M(G, n)$ to the unique unital normal $\dagger$-homomorphism $\pi$ such that $\pi(\delta_g) = p(g)$. The second bijection maps $\pi$ to the function $F_p\: \Q_n \to `V_G$ that is defined by $$F_p(\CC^n, \CC_g) = \{v \in L(\CC^n, \CC_g) \mid v = v \pi(\delta_g)\} = L(\CC^n, \CC_g) \pi(\delta_g) = L(\CC^n, \CC_g)p(g)$$ \cite[Theorem~6.3]{zbMATH07287276}. Altogether, we conclude that $p \mapsto F_p$ is a bijection
\begin{align*}
V_{M(G,n)} \to \cat{qSet}(\Q_n, V_G) = \cat{qGph}(Q_n, G).
\end{align*}
To show that this bijection is an isomorphism, we reason that for all vertices $p_1$ and $p_2$ of $M(G,n)$,
 \begingroup \makeatletter\tagsleft@false\makeatother
\begin{align*}
&
p_1(g_1)p_2(g_2) = 0_n \text{ for all } g_1 \not \sim g_2
\\ & \EV
L(\CC^n, \CC_{g_1})p_1(g_1)p_2(g_2)L(\CC_{g_2}, \CC^n) = 0 \text{ for all } g_1 \not \sim g_2
\\ & \EV
F_{p_1}(\CC^n, \CC_{g_1})F^\dagger_{p_2}(\CC_{g_2}, \CC^n) = 0 \text{ for all } g_1 \not \sim g_2
\\ & \EV
(F_{p_1} \circ F_{p_2}^\dagger)(\CC_{g_2}, \CC_{g_1}) = 0 \text{ for all } g_1 \not \sim g_2
\\ & \EV
`e_G(\CC_{g_2}, \CC_{g_1}) \geq (F_{p_1} \circ F_{p_2}^\dagger)(\CC_{g_2}, \CC_{g_1}) \text{ for all } g_1 \text{ and } g_2
\\ & \EV `e_G \geq F_{p_1} \circ F_{p_2}^\dagger.\tag*{\qedsymbol}\popQED
\end{align*}\endgroup 
\end{proof}

\begin{theorem}\label{J}
The following are equivalent for finite simple graphs $G$ and $H$:
\begin{enumerate}
\item there is a winning quantum strategy for the $(G,H)$-homomorphism game,
\item $[G, H] \neq K_0$, where $K_0$ is the empty graph.
\end{enumerate}
\end{theorem}

\begin{proof}
For each positive integer $n$, we calculate that
\begin{align*}
\cat{qGph}(Q_n, [G, H])
& 
\iso \cat{qGph}(Q_n \boxprod G, H)
\\ & \iso 
\cat{qGph}(G \boxprod Q_n, H)
\\ & \iso
\cat{qGph}(G, [Q_n, H])
\\ & \iso 
\cat{Gph}(G, \cat{qGph}(K_1, [Q_n, H])) 
\\ & \iso
\cat{Gph}(G, \cat{qGph}(K_1 \boxprod Q_n, H))
\\ & \iso
\cat{Gph}(G, \cat{qGph}(Q_n, H))
\\ & \iso
\cat{Gph}(G, M(H,n)).
\end{align*}
The first, second, and third bijections come from the closed symmetric monoidal structure of $\cat{qGph}$ (section \ref{section 3}). The fourth and fifth bijections then come from \cref{H} and \cref{F}, respectively. Finally, the sixth bijection comes from the $\cat{Gph}$-enriched closed symmetric monoidal strucutre of $\cat{qGph}$ (section~\ref{section 4}), and the seventh bijection comes from \cref{I}.

We prove the theorem via a sequence of property equivalences. First, there is a winning quantum strategy for the $(G,H)$-homomorphism game iff the set $\cat{Gph}(G, M(H,n))$ is nonempty for some positive integer $n$ by \cite[Theorem~2.8]{zbMATH06555288}. Second, for all positive integers $n$, the set $\cat{Gph}(G, M(H,n))$ is nonempty iff the set $\cat{qGph}(Q_n, [G, H])$ is nonempty. Third, the set $\cat{qGph}(Q_n, [G, H])$ is empty for all positive integers $n$ iff $[G, H] = K_0$ because $K_0$ is the unique quantum graph whose vertex quantum set is empty.
\end{proof}

\section{Homomorphisms as quantum operations}\label{section 6}

In this section, we identify homomorphisms between finite quantum graphs with certain quantum operations, confirming that \cref{A} is, for finite quantum graphs, a special case of \cite[Definition~8.5]{zbMATH07388954}. The latter definition is a quantum generalization of substochastic graph homomorphisms, whereas \cref{A} is a quantum generalization of graph homomorphisms in the standard sense.
Much of this section revisits results from \cite{zbMATH07809225}, which we present in a less general and more elementary way.

Any finite physical system, such as a computer with both quantum and classical registers, may be modeled by a finite-dimensional von Neumann algebra. Up to isomorphism, finite-dimensional von Neumann algebras are exactly unital $\dagger$-algebras of matrices, so in this section, we will work explicitly with this class, following \cite{zbMATH07388954}. We call such algebras \emph{matrix $\dagger$-algebras}. A \emph{quantum operation} between two finite physical systems is then defined to be a trace-nonincreasing completely positive map $\varphi\: \M \to \N$, where $\M \subseteq M_m(\CC)$ and $\N \subseteq M_n(\CC)$ are matrix $\dagger$-algebras \cite{zbMATH05995689}. 

Recall that a linear map $\varphi\: M_m(\CC) \to M_n(\CC)$ is completely positive and trace-nonincreasing iff it has a \emph{Kraus decomposition}
$$
\varphi(a) = \sum_{i = 1}^k v_i a v_i^\dagger, \qquad \sum_{i = 1}^k v_i^\dagger v_i \leq 1_m,
$$
where $v_1, \ldots, v_k \in M_{n \times m}(\CC)$ \cite[Theorem~8.1]{zbMATH05835673}. Furthermore, although the Kraus decomposition of $\varphi$ is not unique, the span of the operators $v_1, \ldots, v_k$ does not depend on their choice \cite[Theorem~8.2]{zbMATH05835673}. These elementary facts generalize to all trace-nonincreasing completely positive maps $\varphi\: \M \to \N$, as we now verify (cf.~\cite[section~7]{zbMATH07856619}).

\begin{proposition}\label{K}
Let $\M \subseteq M_m(\CC)$ and $\N \subseteq M_n(\CC)$ be matrix $\dagger$-algebras, and let $\varphi\: \M \to \N$ be a trace-nonincreasing complete positive map. Then,
\begin{enumerate}[ref = \thelemma(\arabic*)]
\item\label[proposition]{K1} there exist matrices $v_1, \ldots, v_k \in M_{n \times m}(\CC)$ such that
$$
\varphi(a) = \sum_{i = 1}^k v_i a v_i^\dagger, \qquad \sum_{i = 1}^k v_i^\dagger v_i \leq 1_m,
$$
\item\label[proposition]{K2} the subspace
$$
\T_\varphi : = \mathrm{span}\{b' v_i a' \mid a' \in \M', \, b' \in \N', \, 1 \leq i \leq k\} \subseteq M_{n \times m} (\CC)
$$
does not depend on the choice of the matrices $v_1, \ldots, v_k$,
\item\label[proposition]{K3} $\varphi$ has a Kraus decomposition in the sense of claim~1 such that
$$
\T_\varphi = \mathrm{span}\{v_i \mid 1 \leq i \leq k\},
$$
\item\label[proposition]{K4} for all $x_1, x_2 \in \CC^m$, the following are equivalent:
\begin{enumerate}
\item $\tr(\varphi(x_1 x_1^\dagger)\varphi(x_2x_2^\dagger))= 0$,
\item $x_1^\dagger (\T_\varphi^\dagger \cdot \T_\varphi) x_2 = \{0\}$,
\end{enumerate}
\item\label[proposition]{K5} for all $x_1 \in \CC^m$ and $x_2 \in \CC^n$, the following are equivalent:
\begin{enumerate}
\item $\tr(x_2 x_2^\dagger\varphi(x_1 x_1^\dagger)) = 0$,
\item $x_2^\dagger \T_\varphi x_1 = \{0\},$
\end{enumerate}
\item\label[proposition]{K6} the pushfoward \cite{zbMATH07388954} along $\varphi$ of a quantum relation $\R$ on $\M$ is the quantum relation $\overrightarrow{R} = \T_\varphi \cdot \R \cdot \T_\varphi^\dagger$ on $\N$,
\item\label[proposition]{K7} the pullback \cite{zbMATH07388954} along $\varphi$ of a quantum relation $\S$ on $\N$ is the quantum relation
$
\overleftarrow{S} = \T_\varphi^\dagger \cdot \S \cdot \T_\varphi
$
on $\M$.
\end{enumerate}
\end{proposition}

\begin{proof}
Let $\pi\:M_m(\CC) \to \M$ be the completely positive map
$$
\pi(a) = \sum_{j = 1}^r p_j a p_j,
$$
where $p_1, \ldots, p_r$ are the minimal central projections of $\M$. We observe that $\pi(a) = a$ for all $a \in \M$ and that $\mathrm{tr}(\pi(a)) = \tr(a)$ for all $a \in M_m(\CC)$. The map $\varphi \circ \pi\: M_m(\CC) \to \N$ is completely positive and trace-nonincreasing, so it has a Kraus decomposition $v_1, \ldots v_k \in M_{n \times m}(\CC)$ by \cite[Theorem~8.1]{zbMATH05835673}. This is also a Kraus decomposition for $\varphi$ because $\varphi(a) = (\varphi \circ \pi)(a)$ for all $a \in \M$. We have proved claim 1.

Let $\varphi(a) = \sum_{i=1}^l w_i a w_i^\dagger$ be another Kraus decomposition of $\varphi$. Then,
$$
\sum_{i = 1}^k \sum_{j = 1}^r v_i p_j a p_j v_i^\dagger = (\varphi \circ \pi)(a) = \sum_{i = 1}^l \sum_{j = 1}^r w_i p_j a p_j w_i^\dagger
$$
for all $a \in M_m(\CC)$. In other words, the matrices $v_i p_j$, for $1 \leq i \leq k$ and $1 \leq j \leq r$, are a Kraus decomposition of $\varphi \circ \pi$ and the matrices $w_i p_j$, for $1 \leq i \leq k$ and $1 \leq j \leq r$, are a Kraus decomposition of $\varphi \circ \pi$ too. It follows by \cite[Theorem~8.2]{zbMATH05835673} that
$$
\mathrm{span}\{v_ip_j \mid 1 \leq i \leq k, \, 1 \leq j \leq r \} = \mathrm{span}\{w_ip_j \mid 1 \leq i \leq l, \, 1 \leq j \leq r \}.
$$
Since $p_j \in \M'$ for all $1 \leq j \leq r$, we conclude that 
\begin{align*}
\mathrm{span}\{b'v_i a' \mid{}& a' \in \M', \, b' \in \N', \,  1 \leq i \leq k\}
\\ & =
\mathrm{span}\{b'v_i p_j a' \mid a' \in \M', \, b' \in \N',\,  1 \leq i \leq k, \, 1 \leq j \leq r\}
\\ & =
\mathrm{span}\{b'w_i p_j a' \mid a' \in \M', \, b' \in \N',\, 1 \leq i \leq l, \, 1 \leq j \leq r\}
\\ & =
\mathrm{span}\{b'w_i a' \mid a' \in \M', \, b' \in \N',\, 1 \leq i \leq l\}.
\end{align*}
Indeed, $p_j a' \in \M'$ for all $a' \in \M'$ and $1 \leq j \leq r$, and $a' = \sum_{j=1}^r p_j a'$ for all $a' \in \M'$. We have proved claim 2.

Let $u_1, \ldots, u_s \in \M'$ and $w_1, \ldots, w_t \in \N'$ be bases of unitaries for $\M'$ and $\N'$, respectively. It follows that
$$
\varphi(a) = \sum_{h = 1}^s \sum_{i = 1}^k \sum_{j =1}^t \left(\frac{w_j v_i u_h}{\sqrt{st}}\right) a \left(\frac{w_j v_i u_h}{\sqrt{st}}\right)^\dagger,
$$
$$
\sum_{h = 1}^s \sum_{i = 1}^k \sum_{j =1}^t \left(\frac{w_j v_i u_h}{\sqrt{st}}\right)^\dagger \left(\frac{w_j v_i u_h}{\sqrt{st}}\right) \leq 1,
$$
$$
\T_\varphi = \mathrm{span}\{w_j v_i u_h \mid 1 \leq h \leq s,\, 1 \leq i \leq k, \, 1 \leq j \leq t\}.
$$
We have proved claim 3.

Let $\varphi(a) = \sum_{i=1}^k v_i a v_i^\dagger$ be a Kraus decomposition of $\varphi$ such that $\T_\varphi = \mathrm{span}\{v_i \mid 1 \leq i \leq k\}$, as in claim 3. The standard argument \cite[section~3]{zbMATH06727895} that $\T_\varphi^\dagger \cdot \T_\varphi$ is the confusability graph of $\varphi$ when $\M = M_m(\CC)$ and $\N = M_n(\CC)$ applies in the general case as well. We reproduce it here, calculating that
\begin{align*}
\tr(\varphi(x_1 x_1^\dagger) \varphi(x_2 x_2^\dagger))  & = \sum_{i =1}^k \sum_{j=1}^k \tr(v_i x_1 x_1^\dagger v_i^\dagger v_j x_2 x_2^\dagger v_j^\dagger)
\\ & =
\sum_{i =1}^k \sum_{j=1}^k   x_1^\dagger v_i^\dagger v_j x_2 x_2^\dagger v_j^\dagger v_i x_1
=
\sum_{i =1}^k \sum_{j=1}^k |x_1^\dagger v_i^\dagger v_j x_2|^2.
\end{align*}
Thus, $\tr(\varphi(x_1 x_1^\dagger) \varphi(x_2 x_2^\dagger)) =0$ iff $x_1^\dagger v_i^\dagger v_j x_2 = 0$, for all $1 \leq i, j \leq k$, which is equivalent to the condition $x_1^\dagger (\T_\varphi^\dagger \cdot \T_\varphi)x_2 = 0$. We have proved claim 4.

The proof of claim 5 is entirely similar. In this instance, we calculate that
\begin{align*}
\tr(x_2 x_2^\dagger\varphi(x_1 x_1^\dagger))
=
x_2^\dagger \varphi(x_1 x_1^\dagger) x_2
=
\sum_{i =1}^k x_2^\dagger v_i x_1 x_1^\dagger v_i^\dagger x_2
=
\sum_{i =1}^k |x_2^\dagger v_i x_1|^2.
\end{align*}

Claims 6 and 7 follow immediately from \cite[Theorem~8.4]{zbMATH07388954} and \cite[Theorem~9.2]{zbMATH07388954}, respectively.
\end{proof}

Quantum operations are a quantum generalization of substochastic maps. Indeed, if $\M$ and $\N$ are commutative, then $\M \iso \CC^X$ and $\N \iso \CC^Y$ for some finite sets $X$ and $Y$, and the quantum operations $\varphi\: \M \to \N$ are known to be in canonical one-to-one correspondence with substochastic maps $X \to Y$, i.e., with conditional subprobability distributions $p(y|x)$ with $x \in X$ and $y \in Y$. Quantum operations model probabilistic quantum channels where outputs may be discarded; the quantity $\tr(x_2 x_2^\dagger\varphi(x_1 x_1^\dagger))$ is the probability of finding the target system in state $x_2$ after preparing the source system in state $x_1$.

Quantum relations are of course a quantum generalization of relations. By design, if $\M$ and $\N$ are commutative, then $\M \iso \CC^X$ and $\N \iso \CC^Y$ for some finite sets $X$ and $Y$, and the quantum relations $\R\: \M \to \N$ are in canonical one-to-one correspondence with relations $X \to Y$. Quantum relations model possibilistic quantum channels where outputs may be discarded; the quantum relation $x_2^\dagger \R x_1 \subseteq \CC$ records the possibility or impossibility of finding the target system in state $x_2$ after preparing the source system in state $x_1$ according to whether it is equal to $\CC$ or to $\{0\}$.

\cref{K} yields a functor from the category $\cat{QO}$ of matrix $\dagger$-algebras and quantum operations to the category $\cat{QR}$ of matrix $\dagger$-algebras and quantum relations. This functor maps each matrix $\dagger$-algebra $\M$ to itself and maps each quantum operation $\varphi\: \M \to \N$ to the quantum relation $\T_\varphi$ of \cref{K2}; it is a functor by \cref{K3}. In effect, this functor maps each probabilistic quantum channel $\varphi$ to the corresponding possibilistic quantum channel $\T_\varphi$, which retains only the possibility or impossibility of each transition as its data; this is \cref{K5}. The quantum confusability graph \cite{zbMATH06727895} can then be simply computed from this possibilistic quantum channel $\T_\varphi$ as $\T_\varphi^\dagger \cdot \T_\varphi$; this is \cref{K4}.

The functor $\cat{QO} \to \cat{QR}$ that is given by $\varphi \mapsto \T_\varphi$ is full but not faithful, and it is furthermore strong symmetric monoidal for the obvious tensor product structures on the two categories. The verification of the many details that are implicit in these claims is laborious but routine. \cref{K} demonstrates that the quantum relation $\T_\varphi$ not only determines the zero-error capacity of $\varphi$ but also whether or not $\varphi$ is a homomorphism of quantum graph structures, as we now observe.

\begin{proposition}
Let $\M \subseteq M_m(\CC)$ and $\N \subseteq M_n(\CC)$ be matrix $\dagger$-algebras, and let $\varphi\: \M \to \N$ be a trace-nonincreasing completely positive map. Further, let $\R \subseteq M_m(\CC)$ and $\S \subseteq M_n(\CC)$ be quantum relations on $\M$ and $\N$, respectively. The following are equivalent:
\begin{enumerate}[ref = \thelemma(\arabic*)]
\item $\varphi$ is a CP morphism \cite[Definition 8.5]{zbMATH07388954} from $(\M, R)$ to $(\N, S)$,
\item $\T_\varphi \cdot \R \cdot \T_\varphi^\dagger \subseteq \S$.
\end{enumerate}
\end{proposition}

\begin{proof}
This proposition follows immediately from \cref{K6}.
\end{proof}

As Weaver observes in \cite[section~9]{zbMATH07388954}, the condition $\R \subseteq \T_\varphi^\dagger \cdot \S \cdot \T_\varphi$ defines an alternative notion of a CP morphism, which is strictly weaker when $\varphi$ is trace-preserving. We now prove that when $\varphi\: \M \to \N$ is a trace-preserving $\dagger$-cohomomorphism, these two conditions are equivalent and define exactly the graph homomorphisms of \cref{A}.

\begin{definition}[cf.~{\cite[Definition~2.7]{zbMATH06936038}}]\label{L}
Let $\M \subseteq M_m(\CC)$ and $\N \subseteq M_n(\CC)$ be matrix $\dagger$-algebras. Then, a \emph{cohomomorphism} $\varphi\: \M \to \N$ is a linear map such that $\varphi^\dagger\: \N \to \M$ is an algebra homomorphism or, equivalently, such that
$$
\mu \circ \varphi = (\varphi \otimes \varphi) \circ \nu,
$$
where the linear maps $\mu\: \M \to \M \otimes \M$ and $\nu\: \N \to \N \otimes \N$ are the adjoints of the multiplication maps for $\M$ and $\N$, respectively.
It is a $\dagger$-cohomomorphism if furthermore $\varphi(a^\dagger) = \varphi(a)^\dagger$ for all $a \in \M$.
\end{definition}

The adjoint $\varphi^\dagger$ in \cref{L} is computed for the Hilbert-Schmidt inner products on $\M$ and $\N$. Any $\dagger$-cohomomorphism $\varphi$ is completely positive since its adjoint $\varphi^\dagger$ is a $\dagger$-homomorphism and hence completely positive.

\begin{remark}\label{M}
When $\M \subseteq M_m(\CC) \iso L(\CC^m)$ is in standard form \cite{zbMATH03476099}, $m = \dim \M$, and the trace on $M_m(\CC)$ restricts to the \emph{adjusted trace} \cite[Remark~1.4]{zbMATH07990859} or \emph{Markov trace} \cite[section~7]{zbMATH07856619} on $\M$, which normalizes to the unique tracial $\delta$-form on $\M$ \cite{zbMATH01742868}. In this case, $\mu\: \M \to \M \otimes \M$ is the comultiplication of $\M$ as a special unitary $\dagger$-Frobenius involution monoid \cite[Theorem~4.6]{zbMATH05909150}.
\end{remark}

To establish a connection between Weaver's CP morphisms \cite[Definition~8.5]{zbMATH07388954} and our homomorphisms of quantum graphs (\cref{A}), we show that the assignment $\varphi \mapsto \T_\varphi$ in \cref{K2} is a one-to-one correspondence between trace-preserving $\dagger$-cohomomorphisms and quantum functions \cite{arXiv:1101.1694}.

\begin{lemma}\label{N}
Let $\M \subseteq M_m(\CC)$ and $\N \subseteq M_n(\CC)$ be matrix $\dagger$-algebras, and let $\varphi\: \M \to \N$ be a trace-preserving $\dagger$-cohomomorphism. Then, 
$$\T_\varphi = \{v \in M_{n \times m}(\CC) \mid b v = v \varphi^\dagger(b) \text{ for all } b \in \N\}.$$
\end{lemma}

\begin{proof}
The $\dagger$-homomorphism $\varphi^\dagger\: \M \to \N$ is unital since $\varphi$ is trace-preserving. By \cite[Proposition~3.1]{arXiv:1101.1694}, there exist matrices $u_1, \ldots, u_r \in M_{n \times m}(\CC)$ such that $b u_i = u_i \varphi^\dagger(b)$ and
$
\varphi^\dagger(b) = \sum_{i =1}^r u_i^\dagger b u_i
$
for all $b \in \N$. Enumerating the minimal central projections $p_1, \ldots, p_s$ and $q_1, \ldots, q_t$ of $\M$ and $\N$, respectively, we define $v_{ijk} = q_k u_i p_j$ and obtain
$$
\varphi^\dagger(b) = \sum_{i=1}^r \sum_{j = 1}^s \sum_{k= 1}^t v_{ijk}^\dagger b v_{ijk}
$$
for all $b \in \N$, since $\sum_{j=1}^s p_j a p_j = a$ for $a \in \M$ and $\sum_{k = 1}^t q_k b q_k = b$ for $b \in \N$.

The standard calculation
$$
\tr(a \varphi^\dagger(b)) = \sum_{i=1}^r \sum_{j = 1}^s \sum_{k= 1}^t \tr(av_{ijk}^\dagger b v_{ijk}) = \sum_{i=1}^r \sum_{j = 1}^s \sum_{k= 1}^t \tr( v_{ijk} a v_{ijk}^\dagger b)
$$
now shows that
$$
\varphi(a) = \sum_{i=1}^r \sum_{j = 1}^s \sum_{k= 1}^t v_{ijk} a v_{ijk}^\dagger
$$
for all $a \in \M$. Since $\varphi^\dagger$ is unital, we have that $ \sum_{i=1}^r \sum_{j = 1}^s \sum_{k= 1}^t v_{ijk}^\dagger v_{ijk} = 1_m$. The matrices $v_{ijk}$ certainly satisfy $b v_{ijk} = v_{ijk} \varphi^\dagger(b)$, so we conclude that $$\T_\varphi \subseteq \F : = \{v \in M_{n \times m}(\CC) \mid b v = v \varphi^\dagger(b) \text{ for all } b \in \N\}.$$

We prove the opposite inclusion $\F \subseteq \T_\varphi$ by reasoning that
$$
\F = \F \cdot \M' \subseteq \F \cdot \T_\varphi^\dagger \cdot \T_\varphi \subseteq \F \cdot \F^\dagger \cdot \T_\varphi \subseteq \N' \cdot \T_\varphi = \T_\varphi,
$$
where $\M' \subseteq \T_\varphi^\dagger \cdot \T_\varphi$ because $1_m \in \T_\varphi^\dagger \cdot  \T_\varphi$ and $\F \cdot \F^\dagger \subseteq \N'$ because $$b v_1 v_2^\dagger = v_2 \varphi^\dagger(b) v_2^\dagger = v_1 v_2^\dagger b$$ for all $v_1, v_2 \in \F$ and $b \in \N$.
Therefore $\T_\varphi = \F$, as claimed.
\end{proof}

We conclude the article by proving that, in the context of finite quantum graphs, our homomorphisms $G \to H$ correspond to those trace-preserving $\dagger$-cohomomorphisms $\ell(\V_G) \to \ell(\V_H)$ that are CP morphisms in either of Weaver's senses \cite[section~9]{zbMATH07388954}, which become equivalent.

\begin{theorem}\label{O}
Let $\M\subseteq M_m(\CC)$ and $\N\subseteq M_n(\CC)$ be matrix $\dagger$-algebras, and let $\R \subseteq M_m(\CC)$ and $\S \subseteq M_n(\CC)$ be quantum relations on $\M$ and $\N$, respectively. We have a canonical one-to-one-to-one correspondence between the following:
\begin{enumerate}[ref = \thelemma(\arabic*)]
\item\label[theorem]{O1} quantum relations $\F$ from $\M$ to $\N$ such that
$$
\M' \subseteq \F^\dagger \cdot \F, \qquad \F \cdot \F^\dagger \subseteq \N', \qquad \F \cdot \R \subseteq \S \cdot \F,
$$
\item\label[theorem]{O2} trace-preserving $\dagger$-cohomomorphisms $\varphi\: \M \to \N$ such that $\overrightarrow{\R} \subseteq \S$,
\item\label[theorem]{O3} trace-preserving $\dagger$-cohomomorphisms $\varphi\: \M \to \N$ such that $\R \subseteq \overleftarrow{\S}$.
\end{enumerate}
It is given by $\F = \T_\varphi$. The pushforward quantum relation $\overrightarrow{\R}$ and the pullback quantum relation $\overleftarrow{\S}$ are both calculated along $\varphi$ \cite[Definitions~8.2~and~9.1]{zbMATH07388954}.
\end{theorem}

\begin{proof}
Together, \cref{N} and \cite[Theorem~3.5]{arXiv:1101.1694} establish that the assignment $\varphi \mapsto \T_\varphi$ is a one-to-one correspondence between trace-preserving $\dagger$-coho\-mo\-mor\-phisms $\varphi\: \M \to \N$ and quantum relations $\F$ from $\M$ to $\N$ such that $\M' \subseteq \F^\dagger \cdot \F$ and $\F \cdot \F^\dagger \subseteq \N'$. The conditions $\overrightarrow{\R} \subseteq \S$ and $\R \subseteq \overleftarrow{\S}$ are equivalent to $\T_\varphi \cdot \R \cdot \T_\varphi^\dagger \subseteq \S$ and $\R \subseteq \T_\varphi^\dagger \cdot \S \cdot \T_\varphi$, respectively, by \cref{K}, so it remains to observe that
$$
\R \subseteq \T_\varphi^\dagger \cdot \S \cdot \T_\varphi
\qquad \Longleftrightarrow \qquad
\T_\varphi \cdot \R \subseteq \S \cdot \T_\varphi
\qquad \Longleftrightarrow \qquad
\R \subseteq \T_\varphi^\dagger \cdot \S \cdot \T_\varphi.
$$
These equivalences follow directly from $\M' \subseteq \T_\varphi^\dagger \cdot T_\varphi$ and $\T_\varphi \cdot \T_\varphi^\dagger \subseteq \N'$.
\end{proof}

A trace-preserving $\dagger$-homomorphism $\M \to \N$ is a \emph{quantum channel} in the sense of being a trace-preserving completely positive map. In order to provide an interpretation of the morphisms in \cref{O} within quantum information theory, we appeal to the fact that every reflexive, symmetric quantum relation $\R$ on $\M$ is the confusability quantum graph of a quantum channel in the sense of \cref{K4}. This was established in greater generality by Verdon \cite[Proposition~3.12]{zbMATH07809225}, answering a question of Daws \cite[section~6.2]{zbMATH07856619}. We provide a short proof of this fact.
\begin{proposition}\label{P}
Let $\M \subseteq M_m(\CC)$ be a $\dagger$-matrix algebra, and let $\R \subseteq M_m(\CC)$ be a reflexive, symmetric quantum relation on $\M$ \cite[Definition~2.4(d)]{zbMATH06008057}. Then, there exist a positive integer $n$ and a trace-preserving completely positive map $\varphi\: \M \to M_n(\CC)$ such that $\R = \T_\varphi^\dagger \cdot \T_\varphi$.
\end{proposition}

\begin{proof}
In the case $\M = M_m(\CC)$, this proposition is widely known \cite[Lemma~2]{arXiv:0906.2527}. Thus, we obtain a trace-preserving completely positive map $\psi\: M_m(\CC) \to M_n(\CC)$ such that $\T_\psi^\dagger \cdot \T_\psi = \R$. Let $\iota\: \M \to M_m(\CC)$ be the inclusion map, which is also trace-preserving and completely positive. The Kraus decomposition of $\iota$ is given by $\iota(a) = 1_m a 1_m$, so $\T_\iota = \M'$. Let $\varphi = \psi \circ \iota$. The functoriality of the assignment $\varphi \mapsto \T_\varphi$ implies that
$
\T_\varphi^\dagger \cdot \T_\varphi = \T_\iota^\dagger \cdot \T_\psi^\dagger \cdot \T_\psi \cdot \T_\iota = \M' \cdot \R \cdot \M' = \R.
$
\end{proof}

We now interpret the morphisms in \cref{P} from the perspective of quantum information theory.

\begin{remark}\label{Q}
Let $\M$ and $\N$ be matrix $\dagger$-algebras that model finite physical systems, and let $\R$ and $\S$ be quantum relations on $\M$ and $\N$, respectively, that model confusability structures on these systems. Without loss of generality, these confusability structures arise from quantum channels that are modeled by trace-preserving completely positive maps $\chi\: \M \to M_{s}(\CC)$ and $\psi\: \N \to M_t(\CC)$ (\cref{P}); $\R = \T_\chi^\dagger \cdot \T_\chi$ and $\S = \T_\psi^\dagger \cdot \T_\psi$. Recall that a quantum channel confuses two states if there is a non-zero probability of transition between them after transmission through the channel \cite{zbMATH06727895}, as in \cref{K4}.

A quantum channel $\varphi\: \M \to \N$ is a trace-preserving $\dagger$-cohomomorphism iff it is adjoint to a unital $\dagger$-homomorphism iff it never increases adjusted entropy \cite[Theorem~1.1]{zbMATH07990859}. Furthermore, it satisfies the equivalent (\cref{O}) conditions $\R \subseteq \overleftarrow{\S}$ and $\overrightarrow{\R} \subseteq {\S}$ iff
$$
\T_\chi^\dagger \cdot \T_\chi \subseteq \T_\varphi^\dagger \cdot \T_\psi^\dagger \cdot \T_\psi \cdot \T_\varphi = \T_{\psi \circ \varphi}^\dagger \cdot \T_{\psi \circ \varphi}.
$$
Thus, quantum relations $\F$ satisfying \cref{O1} correspond to quantum channels $\varphi\:\M \to \N$ such that
\begin{enumerate}
\item $\varphi$ is deterministic for adjusted entropy \cite{zbMATH07990859},
\item $\psi \circ \varphi$ confuses two states of $\M$ whenever $\chi$ confuses them \cite{zbMATH06727895}.
\end{enumerate}
\end{remark}

Of course, quantum relations $\F$ satisfying \cref{O1} are just homomorphisms $\Phi$ using a different construction of $\cat{qRel}$. Thus, \cref{Q} provides a quantum-information-theoretic gloss of $\cat{qGph}$. We formalize this translation in terms of the noncommutative function spaces
$$
\ell(\X) = \prod_{X \in \At(\X)} L(X).
$$
When $\X$ is finite, $\ell(\X) = \ell^\infty(\X) = \ell^1(\X)$, so we use this notation only to emphasize the coincidence between $\ell^\infty(\X)$ and $\ell^1(\X)$ that occurs in \cref{R}. In the general case, quantum relations are defined on von Neumann algebras, while quantum operations are defined on their preduals.

\begin{corollary}\label{R}
Let $G$ and $H$ be finite quantum graphs. Let $\R_G$ be the quantum relation on $\ell(\V_G)$ that is defined by
$$
\R_G = \bigoplus_{X, Y \in \At(\V_\G)} E_G(X, Y),
$$
and let $\R_H$ be defined similarly. There is a one-to-one correspondence between 
\begin{enumerate}[ref = \thelemma(\arabic*)]
\item homomorphisms $\Phi\: G \to H$ (\cref{A}),
\item trace-preserving $\dagger$-cohomomorphisms $\varphi\: \ell(\V_G) \to \ell(\V_H)$ that are CP morphisms for the quantum relations $\R_G$ and $\R_H$ \cite[Definition 8.5]{zbMATH07388954}.
\end{enumerate}
\end{corollary}

\begin{proof}
This is an immediate consequence of \cref{O} and \cite[Proposition~A.2.2]{zbMATH07828321}, which expresses the equivalence of the two constructions of $\cat{qRel}$.
\end{proof}

\begin{remark}\label{S}
The category $\cat{qRel}$ is dagger compact \cite[Theorem~3.6]{zbMATH07287276}, and this allows a reformulation of \cref{A}. We may equivalently define a quantum graph $G$ to be a quantum set $\V_G$ together with a relation $\bend E_G \: \V_G \times \V_G^* \to \mathbf 1$ such that $\bend E_G = \bend E_{G*} \circ B_{G, G^*}$, where $B_{G, G^*}\: \V_G \times \V_{G}^* \to \V_{G}^* \times V_G$ is the braiding \cite[section~3]{zbMATH07828321}. In this notation, a function $\Phi\: \V_G \to \V_H$ is a homomorphism iff $\bend E_G \leq \bend E_H \circ (\Phi \times \Phi_*)$. Appealing to \cite[Theorems~7.4~and~B.8]{zbMATH07287276}, we conclude that $\cat{qGph}$ is equivalent to the following category:
\begin{enumerate}
\item an object is a hereditarily atomic von Neumann algebra $\M$ with a projection $p \in \M \mathbin{\overline\otimes} \M^{op}$ such that $\sigma(p) = p$, where $\sigma\: \M \mathbin{\overline\otimes} \M^{op} \to \M \mathbin{\overline\otimes} \M^{op}$ is the unital normal $\dagger$-antihomomorphism mapping $a_1 \mathbin{\otimes} a_2$ to $a_2 \mathbin{\otimes} a_1$,
\item a morphism from $(\M, p)$ to $(\N, q)$ is a unital normal $\dagger$-homomorphism $\pi\: \N \to \M$ such that $p \leq (\pi \mathbin{\overline\otimes} \pi)(q)$.
\end{enumerate}
\end{remark}

\appendix

\section{Classical quantum graphs}\label{appendix A}

In this appendix, we prove \cref{H}, which concerns the adjunction between the functors $\Inc\:\cat{Gph} \to \cat{qGph}$ and $\cat{qGph}(K_1, -) \: \cat{qGph} \to \cat{Gph}$ of \cref{G}.

\begin{definition}
A quantum set $\X$ is said to be \emph{classical} if every atom of $X$ is one-dimensional \cite[Proposition~10.6]{zbMATH07287276}. We define $\cat{cSet}$ to be the full subcategory of $\cat{qSet}$ whose objects are classical quantum sets. Similarly, we define $\cat{cGph}$ to be the full subcategory of $\cat{qGph}$ whose objects are quantum graphs $G$ such that $\V_G$ is classical. The functor $\Inc\: \cat{Gph} \to \cat{cGph}$ is an equivalence of categories.
\end{definition}

Clearly, every quantum set $\X$ has a largest classical subset \cite[Definition~2.2]{zbMATH07287276}. Formally, we define the \emph{classical part} of $\X$ to be the quantum set $\Cl(\X)$ whose atoms are exactly the one-dimensional atoms of $\X$. Being a subset of $\X$, the quantum set $\Cl(\X)$ is the range \cite[Definition~3.2]{zbMATH07605379} of the associated \emph{inclusion function} $J\: \Cl(\X) \to \X$ \cite[Definition~8.2]{zbMATH07287276}. We record the familiar universal property of this inclusion function.

\begin{lemma}\label{T}
Let $\X$ and $\Y$ be quantum sets, and let $F\: \Y \to \X$ be a function. If $\Y$ is classical, then $F$ factors uniquely through the inclusion function $J\: \Cl(\X) \to \X$.
$$
\begin{tikzcd}
\Y
\arrow{rd}{F}
\arrow[dotted]{d}[swap]{!}
&
\\
\Cl(\X)
\arrow{r}[swap]{J}
&
\X
\end{tikzcd}
$$
\end{lemma}

\begin{proof}
This lemma follows immediately from the duality between quantum sets and hereditarily atomic von Neumann algebras \cite[Theorem~7.4]{zbMATH07287276} and the characterization of the inclusion function within this duality \cite[Lemma~8.3]{zbMATH07287276}.
\end{proof}

In other words, $\cat{cSet}$ is a coreflective subcategory of $\cat{qSet}$. Similarly, $\cat{cGph}$ is a coreflective subcategory of $\cat{qGph}$, as we now show.

\begin{definition}\label{U}
Let $G$ be a quantum graph. We define the \emph{classical part} of $G$ to be the quantum graph $\Cl(G)$ such that
\begin{enumerate}
\item $\V_{\Cl(G)} = \Cl(\V_G)$,
\item $E_{\Cl(G)} = J^\dagger \circ E_{G} \circ J$,
\end{enumerate}
 where $J\: V_{\Cl(G)} \to \V_G$ is the inclusion function \cite[Definition~8.2]{zbMATH07287276}.
\end{definition}

The inclusion function $J$ is immediately a homomorphism $\Cl(G) \to G$, and it satisfies a universal property that is analogous to \cref{T}.

\begin{lemma}\label{V}
Let $G$ be a graph, let $H$ be a quantum graph, and let $\Phi\: G \to H$ be a homomorphism. There exists a unique homomorphism $\Psi\: G \to \Cl(H)$ that makes the following diagram commute.
$$
\begin{tikzcd}
G
\arrow[dotted]{d}[swap]{\Psi}
\arrow{rd}{\Phi}
&
\\
\Cl(H)
\arrow{r}[swap]{J}
&
H
\end{tikzcd}
$$
Furthermore, $\Phi_1 \sim \Phi_2\: G \to H$ iff $\Psi_1 \sim \Psi_2\: G \to \Cl(H)$, where $J \circ \Psi_i = \Phi_i$.
\end{lemma}

\begin{proof}
By \cref{T}, there exists a unique function $\Psi\: \V_G \to \V_H$ that makes the diagram
$$
\begin{tikzcd}
\V_G
\arrow[dotted]{d}[swap]{\Psi}
\arrow{rd}{\Phi}
&
\\
\V_{\Cl(H)}
\arrow{r}[swap]{J}
&
\V_H
\end{tikzcd}
$$
commute. The function $\Psi$ is a homomorphism because
$$
\Psi^\dagger \circ E_{\Cl(H)} \circ \Psi = \Psi^\dagger \circ J^\dagger \circ E_H \circ J \circ \Psi = \Phi^\dagger \circ E_H \circ \Phi  \geq E_G.
$$

Let $\Psi_1$ and $\Psi_2$ be homomorphisms $G \to \Cl(H)$, and let $\Phi_1$ and $\Phi_2$ be the corresponding homomorphisms $G \to H$. Observing that
$$
\Phi_1^\dagger \circ E_H \circ \Phi_2 = \Psi_1^\dagger \circ J^\dagger \circ E_H \circ J \circ \Psi_2 = \Psi_1^\dagger \circ E_{\Cl(H)} \circ \Psi_2^\dagger,
$$
we conclude that $\Phi_1 \sim \Phi_2$ iff $\Psi_1 \sim \Psi_2$.
\end{proof}

The categories $\cat{cGph}$ and $\cat{Gph}$ are equivalent, so we observe that $\cat{cGph}$ is a coreflective subcategory of $\cat{qGph}$.

\begin{lemma}\label{W}
There is a natural isomorphism $\Cl(G) \iso \Inc(\cat{qGph}(K_1, G))$.
\end{lemma}

\begin{proof}
Let $G$ be a quantum graph.
By \cref{U}, each atom of $\V_{\Cl(G)}$ is a one-dimensional atom of $\V_G$. For each atom $X$ of $\V_{\Cl(G)}$, let $\hat X \:`\{\ast\} \to \V_G$ be the function that is defined by $\hat X(\CC, X) = L(\CC, X)$ with the other components vanishing. Let $$\Phi_G\:\V_{\Cl(G)} \to `\cat{qSet}(`\{\ast\}, G)$$ be the relation that is defined by $\Phi_G(X, \CC_{\hat X}) = L(X, \CC_{\hat X})$ with the other components vanishing. This is an invertible function by \cite[Proposition~4.4]{zbMATH07287276} because $\CC_{\hat X} \neq \CC_{\hat Y}$ whenever $X \neq Y$ and because every function $`\{\ast\} \to \V_G$ is of the form $\hat X$ for some one-dimensional atom $X$ of $\V_G$ by \cite[Theorem~7.4]{zbMATH07287276}.

For all atoms $X$ and $Y$ of $\Cl(G)$, we have that $(J^\dagger \circ E_G \circ J)(X, Y) = E_G(X,Y)$ and similarly $(\Phi_G^\dagger \circ `e \circ \Phi_G)(X,Y) = `e(\CC_{\hat X}, \CC_{\hat Y})$, where $e$ is the adjacency relation of the graph $\cat{qGph}(K_1, G)$. Hence, we may reason that
\begin{align*}
(J^\dagger & \circ E_G \circ J)(X,Y) = 0
\EV E_G(X, Y) = 0 \\ & \EV (\hat Y \circ E_G \circ \hat X)(\CC, \CC) = 0
\\ & \EV \hat Y \circ E_G \circ \hat X \not \geq \Id_{\mathbf 1}
 \EV \hat X \not \sim \hat Y \\ & \EV `e(\CC_{\hat X}, \CC_{\hat Y}) = 0 \EV (\Phi_G^\dagger  \circ `e \circ \Phi_G)(X,Y) = 0.
\end{align*}
Since $X$ and $Y$ are both one-dimensional, we find that $J^\dagger \circ E_G \circ J = \Phi_G^\dagger  \circ `e \circ \Phi_G$. Therefore, $\Phi_G \circ (J^\dagger \circ E_G \circ J) = \Phi_G \circ \Phi^\dagger_G  \circ `e \circ \Phi_G = `e \circ \Phi_G$, so $\Phi_G$ is an isomorphism by \cref{B}. The proof of its naturality is routine.
\end{proof}

\begin{theorem}\label{X}
The functors $\Inc\:\cat{Gph} \to \cat{qGph}$ and $\cat{qGph}(K_1, -) \: \cat{qGph} \to \cat{Gph}$ of \cref{G} are adjoint. This adjunction is enriched over $\cat{Gph}$ in the sense that the natural bijection
$$
\cat{Gph}(G, \cat{qGph}(K_1,H)) \iso \cat{qGph}(\Inc(G), H),
$$
where $G$ is a graph and $H$ is a quantum graph, is an isomorphism.
\end{theorem}

\begin{proof}
We have natural isomorphisms
\begin{align*}
\cat{Gph}(G, \cat{qGph}(K_1,H))
& \iso
\cat{qGph}(\Inc(G), \Inc(\cat{qGph}(K_1,H)))
\\ & \iso
\cat{qGph}(\Inc(G), \Cl(H))
\\ & \iso
\cat{qGph}(\Inc(G), H).
\end{align*}
The first natural isomorphism is induced by the full and faithful symmetric monoidal functor $\Inc\: \cat{Rel} \to \cat{qRel}$ \cite[section~3]{zbMATH07287276}. It preserves adjacency because the adjacency of graph homomorphisms is defined the same way in $\cat{Rel}$ and $\cat{qRel}$. The second natural isomorphism comes from \cref{W}; $\cat{qGph}(\Inc(G),-)$ is a functor $\cat{qGph} \to \cat{Gph}$ because $\cat{qGph}$ is enriched over $\cat{Gph}$ (section~\ref{section 4}). The third natural isomorphism comes from \cref{V}.
Therefore, we obtain the claimed natural isomorphism as a composition of three natural isomorphisms. 
\end{proof}

\bibliographystyle{plain}
\bibliography{refs.bib}

\Addresses

\end{document}